\documentclass[a4paper,fleqn]{cas-sc}

\usepackage[utf8]{inputenc}
\usepackage[T1]{fontenc}

\usepackage[numbers]{natbib}
\usepackage{graphicx}
\usepackage{graphics}
\usepackage[svgnames]{xcolor}
\usepackage{fca}
\usepackage{mathtools}
\usepackage{enumitem}
\usepackage{booktabs}
\usepackage{tikz}
\usetikzlibrary{shapes.symbols, arrows.meta, positioning, fit, backgrounds}
\usepackage{tikzscale}
\usepackage[outline]{contour}   
\contourlength{2.5pt}  
\dimendef\prevdepth=0    
\usepackage{svg}
\usepackage[ruled,vlined]{algorithm2e}
\usepackage{cleveref}
\usepackage{amsmath, amssymb}  
\usepackage{lipsum} 
\usepackage{comment}
\usepackage{siunitx}
\usepackage{amsthm}
\usepackage{bm}
\usepackage{caption,subcaption}
\usepackage{pict2e}
\usepackage{wrapfig}
\usepackage{xcolor}
\usepackage{xfrac} 

\def\tsc#1{\csdef{#1}{\textsc{\lowercase{#1}}\xspace}}
\tsc{WGM}
\tsc{QE}

\newtheorem{theorem}{Theorem}

\newdefinition{rmk}{Remark}
\newproof{pf}{Proof}
\newproof{pot}{Proof of Theorem \ref{thm}}

\newcommand{\old}[1]{{\color{teal} #1}}

\renewcommand{\old}[1]{ }
\newcommand{\commentout}[1]{ }

\newcommand{\R}{\mathbb{R}}

\newcommand{\ie}{i.\,e.\xspace}

\newcommand{\K}{\mathbb{K}}

\newcommand{\N}{\mathbb{N}}
\newcommand{\B}{\mathfrak{B}}
\newcommand{\uB}{\underline{\mathfrak{B}}}
\newcommand{\BB}{\mathfrak{B}}

\newcommand{\LD}{\mathrm{LD}} 

\newcommand{\Add}{\mathrm{Add}} 
\newcommand{\AddLD}{\mathrm{AddLD}} 
\newcommand{\reLD}{\mathrm{reLD}} 

\DeclareMathOperator{\MOD}{mod}
\DeclareMathOperator{\REP}{rep}
\DeclareMathOperator{\REPD}{repd}
\DeclareMathOperator{\POS}{pos}
\DeclareMathOperator{\VEC}{vec}
\DeclareMathOperator{\DIM}{dim}
\DeclareMathOperator{\INC}{inc}
\DeclareMathOperator{\RANK}{rank}
\newcommand{\RR}{\mathcal{R}}
\newcommand{\JJ}{\mathcal{J}}
\newcommand{\MM}{\mathcal{M}}
\newcommand{\CC}{\mathcal{C}}
\newcommand{\BBK}{\BB(\K)}
\newcommand{\UBBK}{\UBB(\K)}
\newcommand{\FMT}{{F\!M}(3)}
\newcommand{\DC}{\mathrel{\dot{\cup}}}

\newcommand{\UBB}{\underline{\BB}}
\newcommand{\PM}{\mathfrak{P}}

\let\cref\Cref

\begin{document}
\shorttitle{DimFlux: Force-Directed Additive Line Diagrams}

\shortauthors{M. Nöhre et~al.}

\title[mode=title]{DimFlux: Force-Directed Additive Line Diagrams}

\author[1,2]{Marcel Nöhre}[orcid=0009-0005-4089-2925]
\ead{noehre@cs.uni-kassel.de}
\author[3]{Dominik Dürrschnabel}[orcid=0000-0002-0855-4185]
\ead{dominik.duerrschnabel@mosbach.dhbw.de}
\author[4]{Bernhard Ganter}[orcid=0000-0003-0767-1379]
\ead{bernhard.ganter@tu-dresden.de}
\author[1,2]{Gerd Stumme}[orcid=0000-0002-0570-7908]
\ead{stumme@cs.uni-kassel.de}
\affiliation[1]{organization={Knowledge \& Data Engineering Group},
	addressline={University of Kassel}, 
	postcode={34121}, 
	city={Kassel},
	country={Germany}}
\affiliation[2]{organization={Interdisciplinary Research Center for Information Systems Design},
	addressline={University of Kassel}, 
	postcode={34121}, 
	city={Kassel},
	country={Germany}}
\affiliation[3]{organization={Baden Wuerttemberg Cooperative State University (DHBW) Mosbach},
	postcode={74821},
	city={Mosbach},
	country={Germany}}
\affiliation[4]{organization={School of Science},
	addressline={TU Dresden}, 
	postcode={01069}, 
	city={Dresden},
	country={Germany}}

\begin{abstract}
The visualization of concept lattices is a central problem in the field of Formal Concept Analysis. Force-directed algorithms, as popular in graph drawing, are a promising approach, treating lattice diagrams as physical models, optimizing node positions based on forces derived from the lattice structure. We build on the work of Zschalig, who, however, limited himself to attribute-additive diagrams. We use a more general additivity, in which both the attributes and the objects contribute to the positions of the concept nodes.

We replace the planarity enhancer used by Zschalig to obtain a starting diagram for force-directed optimization with the DimDraw algorithm, which generates structured order diagrams on its own. The combination results in \textbf{DimFlux}, an algorithm that leverages the advantages of DimDraw but generates additive diagrams in which readability is increased by maximizing the conflict distance between nodes and non-incident edges.
\end{abstract}


\maketitle
\section{Introduction}
\label{sec:introduction}
Line diagrams of concept lattices are a powerful tool for data visualization. However, the ability to produce a layout that truly reflects a lattice's underlying structure remains a specialized skill held by a few experts. To make Formal Concept Analysis more accessible to a broader audience, there is a clear need for a universal drawing algorithm that automates this process.

A line diagram of a concept lattice shows one node in the drawing plane for each formal concept, and a straight, upward-directed line from each concept node to each of its upper covers. Different nodes have to be placed on different positions, and no node shall lie on a non-incident edge. Besides these hard constraints, there are several soft conditions to improve the readability of the drawings, such as maximizing the distance between nodes and non-incident edges, minimizing edge crossings or different slopes, and aligning nodes into vertical chains~\cite{wille1989lattices}. By now, there is no universal algorithm that perfectly balances all these factors. Various algorithms have been developed that focus on subsets of these conditions (see Section~\ref{sec:preliminaries} for details).

In this paper, we present an approach that combines two structural properties that can greatly improve the readability of a line diagram: a) being additive~\cite{wille1989lattices} and b) being realizer-embedded~\cite{Duerrschnabel2023}. 

Ad a): A line diagram of an ordered set $P$ is said to be \emph{additive} if the positions of the nodes are determined as certain sums of a fixed set of vectors, see Section~\ref{sec:additive_diagrams} for a precise definition. 

This construction can greatly improve readability by producing parallelograms in distributive parts of the lattice --- which are surprisingly frequent in real-world data~\cite{abdulla2025rises}. While there are canonical ways for selecting the set of nodes that will be equipped with the initial vectors
(see Section~\ref{sec:additive_diagrams}), the challenge is to define suitable vectors. For \emph{attribute-additive diagrams}, C. Zschalig has developed an algorithm that locally maximizes the conflict distance between nodes and non-incident edges using a force-directed placement model, by treating the lattice as a physical system~\cite{ZschaligFDP}. While this algorithm produces well readable diagrams for small lattices, for larger lattices it becomes trapped in local maxima and produces sub-optimal diagrams. To align with the more general \emph{doubly-additive diagrams}, we will present an extension of this algorithm in Section~\ref{sec:fdp_doubly}.

Ad b): In a line diagram of an ordered set $P$ that is realizer-embedded, the position of the node that is representing an element $p$ of $P$ is  determined by the position of $p$ in the two linear orders of a realizer of a suitable two-dimensional extension, see Section~\ref{sec:additivity_dim_draw} for more details.

Unfortunately, there are concept lattices for which no line diagrams exist that are realizer-embedded and additive at the same time (see Theorems~\ref{thm:additive_realizer_rd} and~\ref{thm:additive_extension}). Therefore, we project the realizer-embedded line diagrams to the closest additive diagram, which serves as the initial layout for the subsequent optimization using the extension of the force-directed model for doubly-additive line diagrams. This algorithm, which we name \textbf{DimFlux}, combines the advantages of both approaches.  

Overall, the contributions of our paper are:
\begin{enumerate}
	\item Extension of Zschalig's force-directed placement approach to doubly-additive line diagrams, 
	\item method for projecting any line diagram to the closest additive one, 
	\item proof of the incompatibility of being additive and being realizer-embedded, 
	\item the DimFlux algorithm, which first generates a realizer-embedded diagram, which is then projected into the space of additive diagrams, where subsequently the conflict distance is maximized.
\end{enumerate}

The paper is organized as follows: after recalling the basics of Formal Concept Analysis and lattice theory in Section~\ref{sec:preliminaries}, we introduce the additive representations used throughout this work in Section~\ref{sec:additive_diagrams} and describe the projection of line diagrams into the additive space in Section~\ref{sec:additive_projection}. We then provide an overview of Zschalig's force-directed placement approach in Section~\ref{sec:fdp_m}, which we extend for doubly-additive diagrams in Section~\ref{sec:fdp_doubly}. Section~\ref{sec:additivity_dim_draw} investigates the additivity of realizer-embedded line diagrams, leading to the introduction of our new algorithm DimFlux for drawing concept lattices in Section~\ref{sec:dim_flux}. Finally, we evaluate DimFlux against state-of-the-art drawing algorithms for concept lattices in Section~\ref{sec:evaluation} and show its limitations in Section~\ref{sec:limit}. Section~\ref{sec:conclusion} concludes our work and discusses future research.

\section{Formal Concept Analysis and line diagrams}
\label{sec:preliminaries}
In the following, we recall the basics of Formal Concept Analysis (FCA) and lattice theory~\cite{GanterFCA2024}, and introduce the notation used throughout this work.

\subsection{Concept lattices}
\label{subsec:concept_lattices}
A \emph{formal context} is defined as a triple $\K := \GMI$, where $G$ and $M$ represent sets of objects and attributes, respectively, and $I \subseteq G \times M$ is a binary relation. For an object $g \in G$ and an attribute $m \in M$, the notation $g \relI m$  (or $(g, m) \in \relI$) indicates that object $g$ has attribute $m$.

For any subset of objects $A \subseteq G$, we define $A'$ as the set of attributes common to all objects in $A$. Dually, for any subset of attributes $B \subseteq M$, $B'$ is the set of all objects that have all attributes in $B$:
$$
\begin{aligned}
    A' &:= \{ m \in M \mid \forall g \in A \colon (g, m) \in I \} \\
    B' &:= \{ g \in G \mid \forall m \in B \colon (g, m) \in I \}\enspace.
\end{aligned}
$$
A \emph{concept lattice} of a formal context $\GMI$ is the set of all formal concepts
$$
\BGMI := \{ (A, B) \mid A \subseteq G, B \subseteq M, A = B', B = A' \}\enspace,
$$
equipped with the partial order
$$
(A_1, B_1) \leq (A_2, B_2)  :\iff A_1 \subseteq A_2\enspace.
$$
The resulting structure $\BVGMI = (\BGMI, \leq)$ is a complete lattice. For any subset of formal concepts $C \subseteq \BGMI$, the supremum (join) and infimum (meet) are given by:
$$
\bigvee_{(A, B) \in C} (A, B) = \left( \biggl( \bigcup_{(A, B) \in C} A \biggr)'', \bigcap_{(A, B) \in C} B \right), \quad \bigwedge_{(A, B) \in C} (A, B) = \left( \bigcap_{(A, B) \in C} A, \biggl( \bigcup_{(A, B) \in C} B \biggr)'' \right)\enspace.
$$
A formal concept $c$ is \emph{join-irreducible} if it cannot be represented as the join of its predecessors, satisfying the condition $c > \bigvee \{ \tilde{c} \in \BBK \mid \tilde{c} < c \}$. Dually, it is \emph{meet-irreducible} if $c < \bigwedge \{ \tilde{c} \in \BBK \mid \tilde{c} > c \}$.

Every object $g \in G$ generates an \emph{object concept} denoted as $\gamma (g) := (\{ g \}'', \{ g \}')$, which is called \emph{irreducible} if it cannot be expressed as the join of other concepts. Dually, every attribute $m \in M$ generates an \emph{attribute concept} denoted as $\mu (m) := (\{ m \}', \{ m \}'')$ that is called irreducible if it cannot be expressed as the meet of other concepts. A formal context can be streamlined into a reduced form that preserves the entire lattice structure while eliminating redundant information. Specifically, a formal context $\K = \GMI$ is called \emph{reduced} if the mapping $\gamma \colon G \to \BBK$ is an injection onto the set of join-irreducible concepts $\JJ(\UBBK)$, and the mapping $\mu \colon M \to \BBK$ is an injection onto the set of meet-irreducible concepts $\MM(\UBBK)$.
For a finite lattice $L$, the context $(\JJ(L), \MM(L), \leq)$ is called its \emph{standard context}; and $\uB(\JJ(L), \MM(L), \leq)$ is isomorphic to $L$.

We denote the smallest element of a finite lattice $L$ by $\bot$ and the largest element by $\top$. The \emph{covering relation} $\prec$ of $L$ is defined by $u \prec v$ if $u < v$ and there is no element $z$ such that $u < z < v$. We then call $u$ a \emph{lower cover} of $v$, and $v$ an \emph{upper cover} of $u$. The \emph{rank} of elements of $L$ is defined recursively: $\RANK(\bot):=0$ and $\RANK(v):=\max\{\RANK(u) \mid u \prec v \} + 1$.  

\subsection{Line diagrams}
\label{subsec:line_diagrams}

Line diagrams of concept lattices are a powerful tool for data visualization.  We can formalize the definition as follows: A line diagram of a lattice $L$ represents each element $u$ of $L$ by a node $\POS(u)$ in the drawing plane $\R^2$. For $u\prec v$ we request that $\POS(u)_y < \POS(v)_y$, and we link both nodes with a straight line. (Here and in the following, we ignore the additional requirement that two nodes should not be positioned at the same position, nor should a node be positioned on a non-adjacent line. This condition will need to be checked on a case-by-case basis.)

A line diagram is thus completely defined by the positions of all nodes, \ie by a mapping $\POS \colon L \to \R^2$. The set of all line diagrams of a lattice $L$ is a subset of the set $(\R^2)^L$ of all mappings from $L$ to $\R^2$, which is isomorphic to $\R^{2|L|}$. It is given by  
$$
\LD(L) := \{\POS\in (\R ^2)^L \mid \forall u,v\in L \colon u \prec v \Rightarrow \POS(u)_y < \POS(v)_y\}\enspace.
$$
$\LD(L)$ is a convex cone, \ie, if $\alpha\in\R_{>0}$ and $\POS_1\in\LD(L)$ and $\POS_2\in\LD(L)$, then $\alpha\cdot\POS_1\in\LD(L)$ and $\POS_1+\POS_2\in\LD(L)$. This convex cone is an open set, as it is the intersection of all open half-spaces of the form $\{\POS\in (\R^2)^L \mid \POS(u)_y < \POS(v)_y\}$, with one open half-space for each couple $u\prec v$ in the covering relation. For convenience, we let $\LD(\K):=\LD(\uB(\K))$ for a context $\K$. 

Computing visually appealing layouts remains a persistent challenge with no universal solution. A drawing must additionally satisfy further structural constraints: it must prevent concept nodes from overlapping with non-incident edges. Furthermore, the readability is refined through several soft conditions, such as minimizing the number of edge crossings and different slopes, while maximizing the length of chains or the distance between nodes and non-incident edges.

While there are several reasonably fast algorithms for generating concept lattices and their neighborhood relations, generating good line diagrams is apparently more difficult. Common algorithms attempt to optimize various criteria: The algorithm of K. Sugiyama et al. determines the $y$-coordinate by a ranking function and then iteratively tries to minimize the number of line crossings between adjacent levels~\cite{sugiyama1981methods}. The algorithm of A. Aeschlimann and J. Schmid minimizes the total length of the edges~\cite{aeschlimann1992drawing}. The algorithm of R. Freese first generates a diagram in 3D by assigning the $z$-coordinates by a rank function and the $x$- and $y$-coordinates by a force-based approach, followed by a projection of this diagram to 2D~\cite{freese2004automated}. Zschalig's algorithm generates attribute-additive diagrams by a force-directed mechanism~\cite{ZschaligFDP}. DimDraw computes a minimal extension of the order relation such that the result has order dimension 2, and then uses, for each node, its position in the two linear extensions of the realizer as coordinates~\cite{Duerrschnabel2023}. The last two algorithms are described in more detail in Sections~\ref{sec:fdp_m} and~\ref{sec:additivity_dim_draw}.

While we usually assume that line diagrams are displayed on a two-dimensional plane, we may consider `line diagrams' in higher dimensions, either for animated visualisations or for theoretical reasons --- as for instance in the previously mentioned algoritm of R. Freese. Such a $d$-dimensional line diagram of a lattice $L$ is then represented by a mapping $\POS \colon L \to \R^d$, which respects the order in the last dimension, \ie, $u\prec v \Rightarrow \POS(u)_d < \POS(v)_d$. We denote the set of all $d$-dimensional line diagram of a lattice $L$ by $\LD^d(L)$.
\section{Additive diagrams}
\label{sec:additive_diagrams}
One approach to addressing the persistent challenge of computing visually appealing layouts is to use additive diagrams~\cite{wille1989lattices}. We illustrate this using a small example that has already been used in~\cite{GanterConflict2004}. Figure~\ref{fig:bungled} shows a $5\times 4$ context of the five dwarf planets and an (admittedly amateurish and --- as we will see --- non-additive) line diagram of its concept lattice.

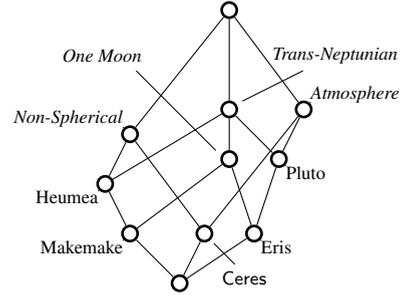
\begin{figure}
    \centering
    \begin{subfigure}[b]{0.625\textwidth}
        \centering
        \begin{cxt}[b]
            \att{Non-Spherical}\att{Atmosphere}\att{Trans-Neptunian}\att{One Moon}
            \obj{xx..}{Ceres}
            \obj{x.xx}{Makemake}
            \obj{.xxx}{Eris}
            \obj{x.x.}{Heumea}
            \obj{.xx.}{Pluto}
        \end{cxt}
    \end{subfigure}%
    \begin{subfigure}[b]{0.375\textwidth}
        \centering
        \begin{tikzpicture}[scale=.3]
            \begin{diagram}\fcaCircleSize{.2}
                \Node (0)(1.0909,12.0)
                \Node (1)(1.0909,7.6364)
                \Node (2)(1.0909,5.4545)
                \Node (3)(4.3636,7.6363)
                \Node (4)(3.2727,5.4545)
                \Node (5)(2.1818,2.1818)
                \Node (6)(-3.2727,6.5454)
                \Node (7)(-4.3636,4.3636)
                \Node (8)(-3.2727,2.1818)
                \Node (9)(0.0,2.1818)
                \Node (10)(-1.0909,0.0)
                \Edge(1)(0)
                \Edge(2)(1)
                \Edge(3)(0)
                \Edge(4)(3)
                \Edge(4)(1)
                \Edge(5)(2)
                \Edge(5)(4)
                \Edge(6)(0)
                \Edge(7)(6)
                \Edge(7)(1)
                \Edge(8)(2)
                \Edge(8)(7)
                \Edge(9)(6)
                \Edge(9)(3)
                \Edge(10)(9)
                \Edge(10)(8)
                \Edge(10)(5)
                \node at (1.75, 0.25){\scriptsize Ceres};
                \put(0.4, 0.25){\line(-0.5, 0.5){0.275}}
                \leftObjbox(8){\scriptsize Makemake}
                \rightObjbox(5){\scriptsize Eris}
                \leftObjbox(7){\scriptsize Heumea}
                \rightObjbox(4){\scriptsize Pluto}
                \leftAttbox(6){\scriptsize Non-Spherical}
                \node at (5.75, 10.0){\scriptsize\itshape \textrm{Trans-Neptunian}};
                \put(1.3, 2.8){\line(-2, -1){0.8}}
                \node at (-4.5, 10.0){\scriptsize\itshape \textrm{One Moon}};
                \put(-0.9, 2.8){\line(2.25, -2.25){1.05}}
                \rightAttbox(3){\scriptsize \textrm{Atmosphere}}
            \end{diagram}
        \end{tikzpicture}
    \end{subfigure}%
    \caption{A small formal context of the five dwarf planets and a poorly designed diagram of its concept lattice.}
    \label{fig:bungled}
\end{figure}

In additive diagrams, the nodes cannot be positioned arbitrarily. Instead, a set of vectors is selected, and each node is assigned a position as a linear combination of these vectors according to specific rules. The motivation for this approach originally stemmed from the observation that such diagrams are often easy to read, especially when drawing distributive lattices. This, in turn, is due to the fact that additive diagrams often contain many parallelograms, which aid readability. 

We will now refine this approach and focus on concept lattices. Let $\GMI$ be a finite formal context and $\BVGMI$ the associated concept lattice for which a line diagram is to be drawn. First, we choose a \emph{set representation} of $\BVGMI$, by which we mean an order embedding in the power set of some set $S$, i.e., a mapping 
$$
\REP \colon \BVGMI \to (\PM(S),\subseteq)\enspace.
$$
Such set representations are easy to find, and it is  well known what the smallest possible size of the set $S$ is, namely the 2-dimension of the embedded ordered set~\cite[Proposition 120]{GanterFCA2024}. 

There are several obvious set representations for a concept lattice $\BVGMI$. One is to choose $S := G$ and $\REP_o(A,B) := A$, another takes $S := M$ and $\REP_a(A,B) := M \setminus B$. We then refer to \emph{object-additive} or \emph{attribute-additive} diagrams, usually under the assumption that $\GMI$ is reduced. Combining these two representations (conveniently assuming $G \cap M = \emptyset$) and choosing $S := G \DC M$, we obtain the \emph{doubly-additive} representation 
$$
\REP(A,B):= A \cup (M \setminus B)\enspace.
$$
When we refer to \emph{additive} diagrams in the following, we always mean doubly-additive diagrams; the cases of attribute-additive and object-additive diagrams can be understood as special cases of this.

In the next step, a vector $\VEC(s) \in \R^2$ is selected for each element $s \in S$, and the concept nodes of the diagram are placed according to the rule
$$
\POS(A,B) := \sum_{s \in \REP(A,B)} \VEC(s)\enspace.
$$
This assigns a place in the plane to each formal concept, and because it is clear which nodes are to be connected by edges, a graph diagram of the concept lattice is determined. For a given context $\K$ and a given set representation $\REP\uB(\K)\to  (\PM(S),\subseteq)$, 
$$
\Add_S(\K) := \left\{\POS\in (\R^2)^{\uB(\K)} \quad\mid\quad \exists \VEC\in (\R^2)^S \colon \forall (A,B)\in\uB(\K) \colon \POS(A,B) = \textstyle\sum_{s \in \REP(A,B)} \VEC(s)\right\}
$$
is the set of all line diagrams that can be created for $L$ based on the set representation. (We just write $\Add(\K)$ if $S$ is clear from the context.)

However, this does not guarantee that a valid order diagram has actually been created. To do this, it must be ensured that $\POS$ is in 
$$
\AddLD(\K) := \Add(\K) \cap \LD(\K)\enspace.
$$
A simple sufficient condition for this is that the vectors $\VEC(s)$, $s \in S$, all have a positive $y$-component.

\bigbreak\par\noindent
\textbf{Remark 1: } Some authors prefer to draw diagrams ``from top to bottom'', i.e., starting with the largest concept $(G, G')$.  To do this, they often use the \emph{dual set representation} 
$$
\REPD(A, B) := B
$$
for which 
$$
(A_1, B_1) \le (A_2, B_2) \iff \REPD(A_1, B_1) \supseteq \REPD(A_2, B_2)\enspace.
$$
Vectors with a \emph{negative} $y$-component are usually used to ensure the order diagram properties. This approach can be easily translated into the one described above. If we set $\VEC(M) := \sum_{m \in M} \VEC(m)$, we obtain the following for an additive placement:
$$
\begin{minipage}{.65\textwidth}
  \begin{eqnarray*}
    \POS(A, B) &=& \sum_{s \in \REP_a(A, B)} \VEC(s) = \sum_{s \in M \setminus B}\VEC(m) \\
    &=& \sum_{s \in M} \VEC(s) - \sum_{s \in B} \VEC(s) = \VEC(M) - \sum_{s \in B} \VEC(s) \\
    &=& \VEC(M) - \sum_{s \in \REPD(A, B)} \VEC(s)\enspace.
  \end{eqnarray*}
\end{minipage}
$$
The difference is therefore merely a change in sign and a shift that moves the largest association element to the origin.

\section{Projections into the space of additive diagrams}
\label{sec:additive_projection}
$\Add(\K)$ is a subspace of $(\R^2)^{\uB(\K)}$, as is easy to understand with the help of the \emph{set representation matrix} (SRM), which is defined as follows: SRM is a $0,1$-matrix which has as many rows as there are formal concepts in $\B(\K)$ and as many columns as there are elements in $S$, and the entry in row $(A, B) \in \B(\K)$ and column $s \in S$ if $1$ iff $s \in \REP(A, B)$ and $0$ otherwise. Figure~\ref{fig:srm} shows the SRM for the context in Figure~\ref{fig:bungled}.

\begin{figure}
  \centering
  \small
    $$
    \left(
    \begin{array}[c]{rrrrrrrrr}
        1& 1& 1& 1& 1& 1& 1& 1& 1 \\
        0& 1& 1& 1& 1& 1& 1& 0& 1 \\
        0& 1& 1& 0& 0& 1& 1& 0& 0 \\
        1& 0& 1& 0& 1& 1& 0& 1& 1 \\
        0& 0& 1& 0& 1& 1& 0& 0& 1 \\
        0& 0& 1& 0& 0& 1& 0& 0& 0 \\
        1& 1& 0& 1& 0& 0& 1& 1& 1 \\
        0& 1& 0& 1& 0& 0& 1& 0& 1 \\
        0& 1& 0& 0& 0& 0& 1& 0& 0 \\
        1& 0& 0& 0& 0& 0& 0& 1& 1 \\
        0& 0& 0& 0& 0& 0& 0& 0& 0 \\
    \end{array}
    \right)
    $$
  \caption{The set representation matrix SRM for Figure~\ref{fig:bungled} (doubly-additive). The rows are labeled by the eleven formal concepts (in lectic order of their intents) and the columns are labeled by the objects and attributes.}
  \label{fig:srm}
\end{figure}

\subsection{Additivity-preserving transformations}
\label{subsec:additivity_preserving_transformations}
To obtain a placement, additionally a mapping $\VEC\colon S\to \R^2$ as described in Section~\ref{sec:additive_diagrams} must be selected. We can write the mapping in the form of a matrix with $|S|$ rows and two columns, where the $s$-th row contains the two components $\VEC(s)_x$ and $\VEC(s)_y$ of the vector $\VEC(s)$. If SRM is now multiplied by this $|S| \times 2$ matrix from the right, the result is a placement according to the rule specified above. This means, for each concept $(A,B)$, that it will be positioned at 
$$
\POS(A, B) = \sum_{s \in \REP(A, B)} \VEC(s)\enspace.
$$
The general additive placements for the selected set representation thus form the image space of this matrix. More precisely, the following applies: $(x_1, y_1), (x_2, y_2), \ldots$ is an additive placement for $\BVGMI$ with the selected set representation iff both column vectors $(x_1, x_2, \ldots )^T$ and $(y_1, y_2, \ldots)^T$ are linear combinations of the columns of SRM. Each additive placement, therefore, consists of two vectors from the image space of the SRM matrix, and conversely, two vectors from this space result in an additive placement. In this sense, additive placements are (arbitrary) pairs of elements of $\textrm{im(SRM)}$, the image space of SRM. Thus $\Add_S(\K) = \textrm{im(SRM)}\times \textrm{im(SRM)}$.

Let us now consider again line diagrams that respect the constraint $p\prec q\Rightarrow \POS(p)_y<\POS(q)_y$. In an additive line diagram this means that if $(A_1, B_1)$ is a lower cover of $(A_2, B_2)$, then $\POS(A_1, B_1)$ must have a smaller $y$-coordinate than $\POS(A_2, B_2)$. This is expressed by the strict linear inequality
$$
\sum_{s \in \REP(A_2, B_2) \setminus \REP(A_1, B_1)} \VEC(s)_y > 0\enspace,
$$
which defines an open half-space of the vector space. The vectors that simultaneously satisfy all these conditions therefore form an open convex cone in the vector space generated by the columns of the SRM matrix. Incidentally, this also applies if one replaces zero with a minimum value in the formula (which might be handy in practice to assure a certain minimal slope of all edges); this only changes the size of the cone.

Even though it may seem quite theoretical at first glance, this algebraic perspective is definitely of practical value. We are dealing with a subspace of a Euclidean vector space and can therefore use methods such as orthogonal projection. The orthogonal projection assigns to each (non-additive) diagram the nearest additive one. To do this, we use the Gram-Schmidt orthogonalization method, for example, to calculate an orthonormal basis for the column space of the SRM, which we can then use to calculate the projection of any placement into the space of additive placements. Program libraries such as \texttt{NumPy} contain all the necessary routines. 

\subsection{First applications}
\label{subsec:applications}
We briefly outline three possible applications, each for a fixed set representation. As an illustration, we use again the concept lattice of Figure~\ref{fig:bungled}, which has already served for this purpose in~\cite{GanterConflict2004}. 

\subsubsection{Is my diagram additive?}
\label{subsubsec:additive?}
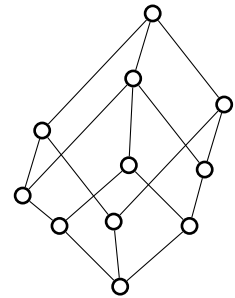
\begin{wrapfigure}{R}{4cm}
    \vspace{-12pt}
    \centering
    \begin{tikzpicture}[scale=.3]
        \begin{diagram}\fcaCircleSize{.2}
            \Node (0)(1.2857,12.0)
            \Node (1)(0.4286,9.1429)
            \Node (2)(0.2381,5.3333)
            \Node (3)(4.4286,8.0000)
            \Node (4)(3.5714,5.1429)
            \Node (5)(2.9048,2.6667)
            \Node (6)(-3.5714,6.8571)
            \Node (7)(-4.4286,4.0000)
            \Node (8)(-2.8095,2.6667)
            \Node (9)(-0.4286,2.8571)
            \Node (10)(-0.1429,0.0)
            \Edge(1)(0)
            \Edge(2)(1)
            \Edge(3)(0)
            \Edge(4)(3)
            \Edge(4)(1)
            \Edge(5)(2)
            \Edge(5)(4)
            \Edge(6)(0)
            \Edge(7)(6)
            \Edge(7)(1)
            \Edge(8)(2)
            \Edge(8)(7)
            \Edge(9)(6)
            \Edge(9)(3)
            \Edge(10)(9)
            \Edge(10)(8)
            \Edge(10)(5)
        \end{diagram}
    \end{tikzpicture}
    \caption{\sffamily\small Additive diagram that has been obtained as an orthogonal projection from the non-additive diagram in Figure~\ref{fig:bungled}.}
    \label{fig-additive-dwarf-planets}
\end{wrapfigure}

Given an arbitrary line diagram, one can determine its additivity by calculating its orthogonal projection into the additive vector space. If it is identical to the original diagram, then it is additive; otherwise, it is not. Using standard linear algebra methods, such as a pseudo-inverse for SRM, vectors can be calculated in the additive case.

The diagram in Figure~\ref{fig-additive-dwarf-planets} is the projection of the diagram from Figure~\ref{fig:bungled}. It is obviously different, so the first diagram was not (doubly) additive.

\subsubsection{Drag and drop}
\label{subsubsec:dragging}
Drag and drop is a convenient method for interacting with line diagrams. But when we manually modify a given additive diagram, we might destroy its additivity. However, this damage can be repaired by applying the orthogonal projection after each step.  There is no guarantee that the projection will respect the order relation. However, if the change is small, there is a good chance that it will not leave the convex cone and thus will project back onto a valid diagram.
    
This can be used, for example, to allow the user of an interactive screen display to move a concept node with the mouse. After each small movement, the projection is calculated, and the diagram is redrawn. Moving one node can thus change the entire diagram. If the user leaves the cone of valid diagrams, the movement is not executed.

The effect is demonstrated in Figure~\ref{fig:dragging}, where the middle atom (the node marked by an outgoing arrow in the leftmost diagram) is moved horizontally to the right in three equal steps. After each such step, the additive approximation is depicted, to which the positions of the nodes in the leftmost diagram are added in gray. Although only the $x$-coordinate of a node is changed, additivity forces other nodes to move as well.

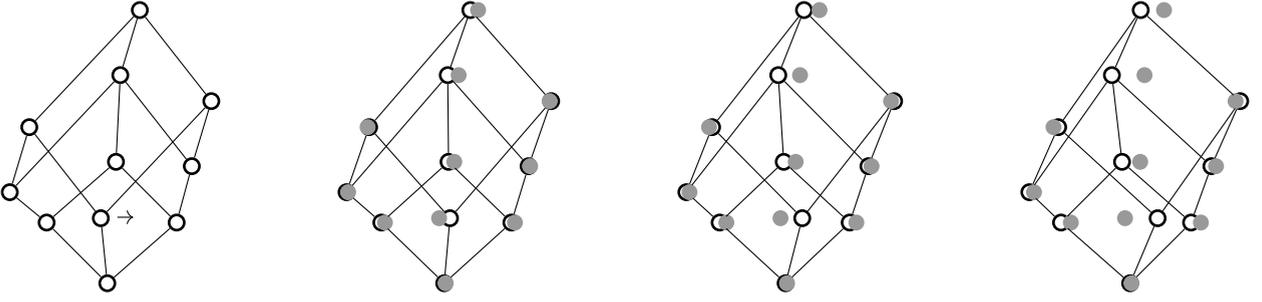
\begin{figure}
  \centering
  \begin{tikzpicture}[scale=.3]
    \begin{diagram}
      \fcaCircleSize{.2}
      \Node (0)(1.2857,12.0)
      \Node (1)(0.4286,9.1429)
      \Node (2)(0.2381,5.3333)
      \Node (3)(4.4286,8.0000)
      \Node (4)(3.5714,5.1429)
      \Node (5)(2.9048,2.6667)
      \Node (6)(-3.5714,6.8571)
      \Node (7)(-4.4286,4.0000)
      \Node (8)(-2.8095,2.6667)
      \Node (9)(-0.4286,2.8571)
      \Node (10)(-0.1429,0.0)
      \draw[->] (.3,2.9) -- (1,2.9);
      \Edge(1)(0)
      \Edge(2)(1)
      \Edge(3)(0)
      \Edge(4)(3)
      \Edge(4)(1)
      \Edge(5)(2)
      \Edge(5)(4)
      \Edge(6)(0)
      \Edge(7)(6)
      \Edge(7)(1)
      \Edge(8)(2)
      \Edge(8)(7)
      \Edge(9)(6)
      \Edge(9)(3)
      \Edge(10)(9)
      \Edge(10)(8)
      \Edge(10)(5)
    \end{diagram}
  \end{tikzpicture}
  \hfill
  \begin{tikzpicture}[scale=.3]
    \begin{diagram}
      \fcaCircleSize{.2}
      \Node (0)(0.9444642857142881,12.0)
      \Node (1)(-0.0491785714285704,9.142857142857148)
      \Node (2)(-0.014273809523809425,5.333333333333334)
      \Node (3)(4.496821428571432,8.000000000000005)
      \Node (4)(3.503178571428575,5.142857142857153)
      \Node (5)(2.7476309523809515,2.6666666666666643)
      \Node (6)(-3.503178571428573,6.857142857142858)
      \Node (7)(-4.496821428571432,4.000000000000005)
      \Node (8)(-2.96665476190476,2.66666666666667)
      \Node (9)(0.04917857142857179,2.857142857142862)
      \Node (10)(-0.2047499999999991,0.0)
      \fill[gray!80] (1.2857,12.0) circle(.35);
      \fill[gray!80] (0.4286,9.1429) circle(.35);
      \fill[gray!80] (0.2381,5.3333) circle(.35);
      \fill[gray!80] (4.4286,8.0000) circle(.35);
      \fill[gray!80] (3.5714,5.1429) circle(.35);
      \fill[gray!80] (2.9048,2.6667) circle(.35);
      \fill[gray!80] (-3.5714,6.8571) circle(.35);
      \fill[gray!80] (-4.4286,4.0000) circle(.35);
      \fill[gray!80] (-2.8095,2.6667) circle(.35);
      \fill[gray!80] (-0.4286,2.8571) circle(.35);
      \fill[gray!80] (-0.1429,0.0) circle(.35);
      \Edge(1)(0)
      \Edge(2)(1)
      \Edge(3)(0)
      \Edge(4)(3)
      \Edge(4)(1)
      \Edge(5)(2)
      \Edge(5)(4)
      \Edge(6)(0)
      \Edge(7)(6)
      \Edge(7)(1)
      \Edge(8)(2)
      \Edge(8)(7)
      \Edge(9)(6)
      \Edge(9)(3)
      \Edge(10)(9)
      \Edge(10)(8)
      \Edge(10)(5)
    \end{diagram}
  \end{tikzpicture}
  \hfill
  \begin{tikzpicture}[scale=.3]
    \begin{diagram}
      \fcaCircleSize{.2}
      \Node (0)(0.6032142857142877,12.0)
      \Node (1)(-0.5269285714285713,9.142857142857151)
      \Node (2)(-0.2872738095238099,5.333333333333335)
      \Node (3)(4.565071428571434,8.000000000000007)
      \Node (4)(3.4349285714285744,5.1428571428571574)
      \Node (5)(2.6111309523809494,2.666666666666664)
      \Node (6)(-3.4349285714285758,6.857142857142857)
      \Node (7)(-4.565071428571434,4.000000000000008)
      \Node (8)(-3.10315476190476,2.6666666666666714)
      \Node (9)(0.5269285714285701,2.8571428571428648)
      \Node (10)(-0.20475000000000076,0.0)
      \fill[gray!80] (1.2857,12.0) circle(.35);
      \fill[gray!80] (0.4286,9.1429) circle(.35);
      \fill[gray!80] (0.2381,5.3333) circle(.35);
      \fill[gray!80] (4.4286,8.0000) circle(.35);
      \fill[gray!80] (3.5714,5.1429) circle(.35);
      \fill[gray!80] (2.9048,2.6667) circle(.35);
      \fill[gray!80] (-3.5714,6.8571) circle(.35);
      \fill[gray!80] (-4.4286,4.0000) circle(.35);
      \fill[gray!80] (-2.8095,2.6667) circle(.35);
      \fill[gray!80] (-0.4286,2.8571) circle(.35);
      \fill[gray!80] (-0.1429,0.0) circle(.35);
      \Edge(1)(0)
      \Edge(2)(1)
      \Edge(3)(0)
      \Edge(4)(3)
      \Edge(4)(1)
      \Edge(5)(2)
      \Edge(5)(4)
      \Edge(6)(0)
      \Edge(7)(6)
      \Edge(7)(1)
      \Edge(8)(2)
      \Edge(8)(7)
      \Edge(9)(6)
      \Edge(9)(3)
      \Edge(10)(9)
      \Edge(10)(8)
      \Edge(10)(5)
    \end{diagram}
  \end{tikzpicture}
  \hfill
  \begin{tikzpicture}[scale=.3]
    \begin{diagram}
      \fcaCircleSize{.2}
      \Node (0)(0.26196428571428865,12.0)
      \Node (1)(-1.0046785714285713,9.142857142857155)
      \Node (2)(-0.5602738095238101,5.333333333333337)
      \Node (3)(4.633321428571436,8.000000000000007)
      \Node (4)(3.366678571428575,5.142857142857162)
      \Node (5)(2.4746309523809487,2.666666666666662)
      \Node (6)(-3.3666785714285763,6.857142857142858)
      \Node (7)(-4.633321428571436,4.0000000000000115)
      \Node (8)(-3.239654761904759,2.6666666666666745)
      \Node (9)(1.0046785714285704,2.857142857142866)
      \Node (10)(-0.20474999999999977,0.0)
      \fill[gray!80] (1.2857,12.0) circle(.35);
      \fill[gray!80] (0.4286,9.1429) circle(.35);
      \fill[gray!80] (0.2381,5.3333) circle(.35);
      \fill[gray!80] (4.4286,8.0000) circle(.35);
      \fill[gray!80] (3.5714,5.1429) circle(.35);
      \fill[gray!80] (2.9048,2.6667) circle(.35);
      \fill[gray!80] (-3.5714,6.8571) circle(.35);
      \fill[gray!80] (-4.4286,4.0000) circle(.35);
      \fill[gray!80] (-2.8095,2.6667) circle(.35);
      \fill[gray!80] (-0.4286,2.8571) circle(.35);
      \fill[gray!80] (-0.1429,0.0) circle(.35);
      \Edge(1)(0)
      \Edge(2)(1)
      \Edge(3)(0)
      \Edge(4)(3)
      \Edge(4)(1)
      \Edge(5)(2)
      \Edge(5)(4)
      \Edge(6)(0)
      \Edge(7)(6)
      \Edge(7)(1)
      \Edge(8)(2)
      \Edge(8)(7)
      \Edge(9)(6)
      \Edge(9)(3)
      \Edge(10)(9)
      \Edge(10)(8)
      \Edge(10)(5)
    \end{diagram}
  \end{tikzpicture}
  \caption{The middle atom is dragged to the right.}
  \label{fig:dragging}
\end{figure}

\subsubsection{Snap to grid}
\label{subsubsec:snap}
Several programs for editing line diagrams allow one to display a grid, similar to graph paper, and then offer the option of moving each concept node to the nearest grid point with a single click. Unfortunately, the result is often disappointing because regularities already found in the diagram can be destroyed.

With additive diagrams, this can be approached systematically in a different way: Instead of moving all concept nodes to the nearest grid point, we just force all vectors $\VEC(s)$ with $s\in S$ to the next grid point. Because the SRM matrix has integer entries only, this forces also all concept nodes to grid points, but this time by retaining the additivity.
\section{Force-directed algorithm for attribute-additive line diagrams}
\label{sec:fdp_m}
Force-directed placement (FDP) for generating graph drawings operates on the principle of modeling a graph as a physical system, in which vertices and, by choice, also edges, are treated as physical elements. By finding a balanced state of attractive and repulsive forces that work on them, the distribution of vertices reflects the structural properties and constraints defined by the underlying physical model~\cite{eades84, Fruchterman1991, Kamada1989}.

C. Zschalig applied this approach to concept lattices by optimizing attribute-additive line diagrams (using the dual set representation) through a three-part process~\cite{ZschaligFDP}. First, he initializes the attribute vectors by exploiting the structural properties of the formal context. Second, he defines a physical model that balances repulsive and attractive forces to maximize the \emph{conflict distance} between edges and non-incident nodes, while a gravitational force ensures an upward-directed diagram that preserves the underlying order relation. Finally, he employs a conjugate gradient optimizer to iteratively refine the forces applied to the attribute vectors until a balanced state is achieved.
In this section, we describe Zschalig's approach in more detail. 

To improve the clarity of the following expressions, we introduce a compact notation for frequently used terms. For any element $e \in G \DC M$, we denote its associated vector as $n = \VEC(e)$. Given the graph $(V, E) = (\BGMI, \prec)$ of a concept lattice, we represent the position of a concept $v \in V$ as $w = \POS (v)$, and the position of an edge $v_1 v_2 \in E$ as $f = \POS_E (v_1 v_2)$. Furthermore, we simplify the notation of the extent $A_v = \operatorname{Ext}(v)$ and the intent $B_v = \operatorname{Int}(v)$. Finally, we introduce a mode operator $\MOD(e)$ to keep our notation compact:
$$
\MOD(e) = 
\begin{cases}
	+1 &, \text{ if } e \in G \\
	-1 &, \text{ if } e \in M\enspace.
\end{cases}
$$
By acting as an indicator for whether an element $e$ is an object or an attribute, this operator prevents redundant dual-case distinctions.

\subsection{Initialization}
\label{subsec:fdp_m_initialization}
The initialization focuses on a layout with minimal edge crossings, as these are one of the main factors for well readable drawings~\cite{Purchase2000, Purchase1997}. To address this, Zschalig's \emph{first planarity condition} (FPC) requires that for $m_i, m_j, m_k \in M$, if $\mu m_k$ lies between $\mu m_i$ and $\mu m_j$ in a linear order, then $\mu m_k > (\mu m_i \wedge \mu m_j)$ must hold to prevent chains from $\mu m_k$ to the bottom concept $c_\bot$ from intersecting with other edges~\cite{Zschalig2007, Zschalig2006, Zschalig2005}. This structural relationship is quantified by the \emph{Sup-Inf distance}, which measures the cardinality of set differences between infimum and supremum:
$$ 
d_{SI}(m_i, m_j) :=
\begin{cases} 
    | (B_{\mu m_i} \cup B_{\mu m_j})'' | - | B_{\mu m_i} \cap B_{\mu m_j} | - 1 &, \text{if } \mu m_i \parallel \mu m_j \\
    0 &, \text{otherwise}
\end{cases} 
$$
for all attributes $m_i, m_j \in M$.

The \emph{planarity enhancer} derives a linear order by treating the weighted graph $\Gamma_{SI} = (M, d_{SI})$ as a physical system. By minimizing the total energy $E_{SI}$ via the forces $F_{SI}$, a layout is reached that reflects the structural distances.
$$
\begin{aligned}
    E_{SI} &= \sum_{m_i, m_j \in M} (|n_i - n_j| - d_{SI}(m_i, m_j))^2 \\
    F_{SI}(n_i) &= -2 \sum_{m_j \in M} \frac{|n_i - n_j| - d_{SI}(m_i, m_j)}{|n_i - n_j|} \cdot (n_i - n_j)\enspace.
\end{aligned}
$$
The final ordering of all $\bigwedge$-irreducibles is extracted by projecting the nodes onto a one-dimensional axis spanned by the two most distant points in the graph. Following that linear order, the coatoms are placed along the parabola $y = -0.09 x^2 - 1.75$, with a consistent spacing of $1.8$ units. Accordingly, the coatoms are placed on the $x$-axis with an offset of $x \in \{ \pm 0.9, \pm 2.7, \ldots \}$ if the number of coatoms is even, or $x \in \{ 0, \pm 1.8, \pm 3.6, \ldots \}$ if it is odd. Finally, the positions for the remaining attributes are determined via \emph{chain decomposition}, where the vector of an attribute is the mean of all $\bigwedge$-irreducibles above it:
$$
n_i = \Delta_i + \frac{1}{|m_j > m_i|} \sum_{m_j > m_i} n_j\enspace,
$$
with $\Delta_i$ representing a small shift, which is necessary if two attributes share the same set of upper neighbors~\cite{ZschaligFDP}.

\subsection{Forces}
\label{subsec:fdp_m_forces}
Since Zschalig's physical model is defined for attribute-additive line diagrams, the forces are not applied to the nodes themselves, but rather to the underlying attribute vectors. The final layout is determined by three forces, with a focus on maximizing the conflict distance 
$$
d(w, f) =
\begin{cases}
    |w_1 - w| ,& \text{if } (w_1 - w) \cdot (w_2 - w_1) > 0 \\
    |w_2 - w| ,& \text{if } (w_2 - w) \cdot (w_2 - w_1) < 0 \\
    \frac{|(w_1 - w) \times (w_2 - w)|}{|w_2 - w_1|} ,& \text{otherwise\enspace,}
\end{cases}
$$
which measures the distance between a concept node $v$ and a non-incident edge $v_1 v_2$. Accordingly, the main force is a \emph{repulsive force}, pushing nodes away from non-incident edges:
$$
F_{rep} = -\nabla E_{rep}, \quad E_{rep} = \sum_{v \in V} \sum_{v_1 v_2 \in \prec, v \not\in v_1 v_2} \frac{1}{d(w, f)}.
$$
To prevent the concept lattice from becoming visually fragmented, a counteracting \emph{attractive force} is applied, which acts like a spring between the concept nodes, pulling them together:
$$
F_{att} = -\nabla E_{att}, \quad E_{att} = \sum_{v_1 v_2 \in \prec} |f|^2\enspace.
$$
Standard FDP models typically operate on undirected graphs, resulting in layouts that lack an orientation. Accordingly, to align with drawing conventions for lattices, a third gravitational force is applied to ensure an upward-directed diagram, pulling the attribute vectors into a safe zone, which is defined by the number of attributes:
$$
F_{grav}(n) = -\frac{\mathrm{d}E_{grav}}{\mathrm{d}n}, \quad \frac{\mathrm{d}E_{grav}(n)}{\mathrm{d}\varphi(n)} = \frac{\sin^2 (\varphi(n)) - \sin^2 (\varphi_0)}{\sin^2 (\varphi(n))} \cdot
\begin{cases}
    \phantom{-}1, \varphi(n) \in [0, \varphi_0] \\
    \phantom{-}0, \varphi(n) \in [\varphi_0, \pi - \varphi_0] \\
    -1, \varphi(n) \in [\pi - \varphi_0, \pi]\enspace,
\end{cases} 
$$
where $\varphi_0 := \frac{\pi}{|M| + 1}$~\cite{ZschaligFDP}.

For the final optimization, Zschalig employs a conjugate gradient approach~\cite{Hestenes1952} to minimize the total energy of the physical system. By stepwise pulling attribute vectors in the direction of the forces, the concept positions are iteratively refined until a local optimum is reached.

\section{Force-directed algorithm for doubly-additive line diagrams}
\label{sec:fdp_doubly}
To adapt Zschalig's force-directed placement approach~\cite{ZschaligFDP} for doubly-additive line diagrams, the Sup-Inf distance must be extended to measure the structural difference for all pairwise combinations of objects and attributes. Furthermore, this representation significantly increases the degree of freedom by working with upward-directed object vectors and downward-directed attribute vectors. Since the underlying physics of the system remains the same, we continue to use the conjugate gradient approach from the original work to iteratively refine the element vectors.

\subsection{Initialization}
\label{subsec:fdp_doubly_initialization}
The original Sup-Inf distance builds on the first planarity condition for the $\bigwedge$-irreducibles. This condition requires that for $m_i, m_j, m_k \in M$, if $\mu m_k$ is situated between $\mu m_i$ and $\mu m_j$ in a linear order, the element $\mu m_k$ must be strictly above their infimum $\mu m_k > (\mu m_i \wedge \mu m_j)$. Otherwise, any chain from $\mu m_k$ to the bottom concept $c_\bot$ would cross either $\mu m_i \mu m_j$ or $\mu m_j \mu m_i$~\cite{Zschalig2007, Zschalig2006, Zschalig2005}. Dually, an \emph{inversed planarity condition} holds for the $\bigvee$-irreducibles. This requires for $g_i, g_j, g_k \in G$, if $\gamma g_k$ is situated between $\gamma g_i, \gamma g_j$, that $\gamma g_k$ must lie strictly below the supremum of the outer pair $\gamma g_k < (\gamma g_i \vee \gamma g_j)$. This prevents any chain from $\gamma g_k$ to the top concept $c_\top$ from crossing the edges $\gamma g_i \gamma g_j$ or $\gamma g_j \gamma g_i$.

Following the attribute-additive approach~\cite{ZschaligFDP}, we generalize the Sup-Inf distance to the doubly-additive representation. To simplify notation, we define a mapping for each element $e_i \in G \DC M$ to its corresponding object or attribute concept $\CC_i \in \BGMI$ as follows:
$$
\CC_i := (A_i, B_i) :=
\begin{cases}
\gamma e_i &, \text{if } e_i \in G \\
\mu e_i &, \text{if } e_i \in M\enspace.
\end{cases}
$$
Using this unified mapping, we extend the original distance measure to any incomparable pair of concepts $\CC_i \parallel \CC_j$ for all elements $e_i, e_j \in G \DC M$. The doubly-additive Sup-Inf distance is defined as:
$$
d_{SI}(e_i, e_j) :=
\begin{cases}
  |(A_i \cup A_j)''| - |A_i \cap A_j| - 1 &, \text{if } \CC_i \parallel \CC_j \text{ and } e_i, e_j \in G \\
  |(B_i \cup B_j)''| - |B_i \cap B_j| - 1 &, \text{if } \CC_i \parallel \CC_j \text{ and } e_i, e_j \in M \\
  \Delta_{\wedge} - \Delta_{\vee} - 1 &, \text{if } \CC_i \parallel \CC_j \text{ and } e_i \in G, e_j \in M \text{ (or vice versa)} \\
  0 &, \text{otherwise\enspace,}
\end{cases}
$$
where the discrepancies $\Delta_{\wedge}$ and $\Delta_{\vee}$ account for the structural differences between extents and intents:
$$
\Delta_{\wedge} = |A_i \cap A_j| - |B_i \cap B_j| \text{ and } \Delta_{\vee} = |(A_i \cup A_j)''| - |(B_i \cup B_j)''|\enspace.
$$
Further extending the planarity enhancer of the original attribute-additive approach, we derive a weighted graph $\Gamma_{SI} = (G \DC M, d_{SI})$. In this, each element $e_i \in G \DC M$ is initially assigned a position on the unit circle:
$$
n_i = (\cos(\varphi(n_i)), \sin (\varphi(n_i))), \text{ where } \varphi(n_i) = \frac{2 \pi i}{|G \DC M|}\enspace.
$$
The graph layout is optimized by pushing each element in the direction of the force $F_{SI}$. In the final balanced state, which minimizes the total energy $E_{SI}$, the Euclidean distances of edges $n_i n_j$ approximate the weights defined by $d_{SI}(e_i, e_j)$:
$$
\begin{aligned}
  E_{SI} &= \sum_{e_i, e_j \in G \DC M} (|n_i - n_j| - d_{SI}(e_i, e_j))^2 \\
  F_{SI}(n_i) &= -2 \sum_{e_j \in G \DC M} \frac{|n_i - n_j| - d_{SI}(e_i, e_j)}{|n_i - n_j|} \cdot (n_i - n_j)\enspace.
\end{aligned}
$$
For the final initialization of the element vectors, we employ the original parabola $y = 0.09 x^2 + 1.75$, which was derived from the positions of $\bigwedge$-irreducibles in planar diagrams of the Boolean lattices $B_4$ and $B_5$~\cite{ZschaligFDP}. By applying the dual transformation $(x, y) \mapsto (-x, -y)$, we obtain the mirrored parabola $y = -0.09 x^2 - 1.75$. Accordingly, this results in a symmetric combination of an upward-directed parabola for the atoms and a downward-directed parabola for the coatoms. The elements are placed with a consistent spacing of $1.8$ units, with respect to the linear order derived from projecting the vertices of the weighted graph $\Gamma_{SI}$ onto a one-dimensional axis spanned by the graph's two most distant points.

The positions for the remaining elements are determined via chain decomposition. The vector of an attribute is defined as the mean of all attribute vectors above it, while the vector of an object is the mean of all object vectors below it:
$$
n_i = \Delta_i + 
\begin{cases}
  \frac{1}{|g_j \leq e_i|} \sum\limits_{g_j \leq e_i} n_j & \text{if } e_i \in G \\
  \frac{1}{|m_j > e_i|} \sum\limits_{m_j > e_i} n_j & \text{if } e_i \in M\enspace,
\end{cases}
$$
where $\Delta_i$ represents a small shift to prevent overlapping concepts in the initial layout, which is necessary whenever two attributes share the same set of upper neighbors or two objects share the same set of lower neighbors.

\subsection{Forces}
\label{subsec:fdp_doubly_forces}
The foundations of this physical model were established by C. Zschalig for attribute-additive line diagrams~\cite{ZschaligFDP}, which represent a special case of the doubly-additive representation. This section combines the original calculations with our necessary extensions to provide a comprehensive overview for the doubly-additive version. Central to the model is again the conflict distance, implemented as a repulsive force balanced by a counteracting attractive force to prevent overly sparse layouts. Since these forces rely on concept nodes and edges, the fundamental energy computation remains unchanged. However, we extend the force calculations to account for the dual contribution of objects and attributes to concept positions. The gravitational force required a larger refactoring, accounting for the larger vector set pointing into two different directions.

\subsubsection{Repulsive Force}
\label{fdp_doubly_repulsive_force}
The repulsive energy includes all pairs of a concept $v$ and a non-incident edge $v_1 v_2$:
$$
E_{rep} = \sum_{v \in V} \sum_{v_1 v_2 \in \prec, v \not \in v_1 v_2} \frac{1}{d(w, f)}\enspace.
$$
To improve the readability of the line diagram, we aim to maximize the \emph{conflict distance} $d(w, f)$. The resulting force acting on an element vector
$$
F_{rep}(n) =
\sum_{v \in V} \sum_{v_1 v_2 \in \prec, v \not \in v_1 v_2} \frac{1}{d(w, f)^2} \cdot \frac{\mathrm{d}d(w, f)}{\mathrm{d}n}
$$
depends on the relative vertical position of the formal concept $w$, which can be below $w_1$, above $w_2$, or in between (see Figure~\ref{fig:scenarios}). The resulting force depends on the contribution of an element to the extents or intents of the concepts, as shown in Figure~\ref{fig:forces}. From these cases, we exclude the two outermost ones, because independent elements and elements that influence all three concepts, and thus shift the entire group, do not affect the conflict distance. Furthermore, four of these cases violate the order, as an object in the extent of $v_1$ must also be present in the extent of $v_2$; a property that dually applies to the attributes.

\begin{figure}
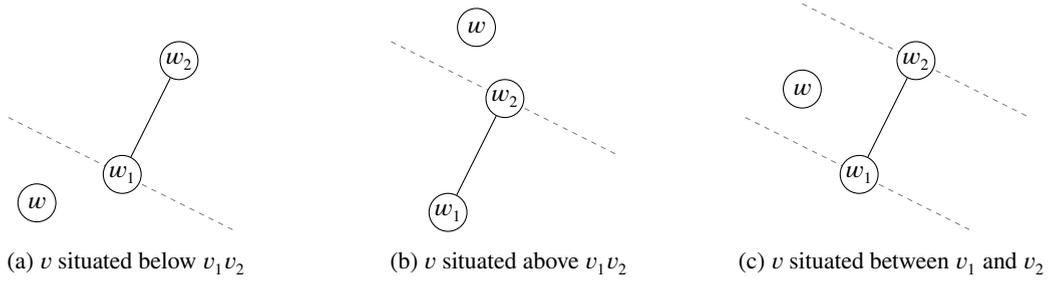

  \centering
    \begin{subfigure}[b]{0.3\textwidth}
      \centering
      \input{figs/fdp/scenario_1}
      \caption{$v$ situated below $v_1 v_2$}
    \end{subfigure}
    \begin{subfigure}[b]{0.3\textwidth}
      \centering
      \input{figs/fdp/scenario_2}
      \caption{$v$ situated above $v_1 v_2$}
    \end{subfigure}
    \begin{subfigure}[b]{0.3\textwidth}
      \centering
      \input{figs/fdp/scenario_3}
      \caption{$v$ situated between $v_1$ and $v_2$}
    \end{subfigure}
  \caption{Possible positions of a concept $v$ relative to a non-incident edge $v_1 v_2$}
  \label{fig:scenarios}
\end{figure}

\begin{figure}
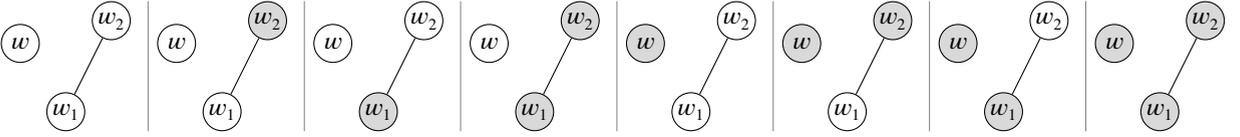

  \centering
  \begin{minipage}{0.125\textwidth}
    \centering
    \input{figs/fdp/force_1}
  \end{minipage}%
  {\color{gray}\vrule width 0.1pt}%
  \begin{minipage}{0.125\textwidth}
    \centering
    \input{figs/fdp/force_2}
  \end{minipage}%
  {\color{gray}\vrule width 0.1pt}%
  \begin{minipage}{0.125\textwidth}
    \centering
    \input{figs/fdp/force_3}
  \end{minipage}%
  {\color{gray}\vrule width 0.1pt}%
  \begin{minipage}{0.125\textwidth}
    \centering
    \input{figs/fdp/force_4}
  \end{minipage}%
  {\color{gray}\vrule width 0.1pt}%
  \begin{minipage}{0.125\textwidth}
    \centering
    \input{figs/fdp/force_5}
  \end{minipage}%
  {\color{gray}\vrule width 0.1pt}%
  \begin{minipage}{0.125\textwidth}
    \centering
    \input{figs/fdp/force_6}
  \end{minipage}%
  {\color{gray}\vrule width 0.1pt}%
  \begin{minipage}{0.125\textwidth}
    \centering
    \input{figs/fdp/force_7}
  \end{minipage}%
  {\color{gray}\vrule width 0.1pt}%
  \begin{minipage}{0.125\textwidth}
    \centering
    \input{figs/fdp/force_8}
  \end{minipage}%
  \caption{Schematic representation of element vector contributions. Gray nodes indicate element vectors contributing to a concept.}
  \label{fig:forces}
\end{figure}

\vspace{3em}
\noindent
(a) \; $(w_1 - w) \cdot (w_2 - w_1) > 0$: \\
For a concept $v$ situated below a non-incident edge $v_1 v_2$, the conflict distance is defined as the absolute difference between its position $w$ and that of the lower edge concept $w_1$:
$$
d(w, f) = |w_1 - w|.
$$
To derive the forces acting on the elements, we apply the unit vector between $w$ and $w_1$. The direction of the force is determined by the contribution of an element to the extents or intents of the concepts, yielding the following gradient:
$$
\frac{\mathrm{d} |w_1 - w|}{\mathrm{d} n} = 
\begin{cases}
  \frac{w_1 - w}{|w_1 - w|} &, \text{if } e \in A_{v_2} \setminus A_v \text{ or } e \in B_{v_1} \setminus B_v \\
  \frac{w - w_1}{|w_1 - w|} &, \text{if } e \in A_v \setminus A_{v_1} \text{ or } e \in B_v \setminus B_{v_1} \\
  0 &, \text{otherwise\enspace.}
\end{cases}
$$
The gradient is applied to the vector $n$ of element $e$, pushing the node away from the non-incident edge if the element influences the individual node but not the full triplet; conversely, it pushes the edge away from the node if the element contributes to an edge concept without influencing the triplet as a whole.

\bigskip
\noindent
(b) \; $(w_2 - w) \cdot (w_2 - w_1) < 0$: \\
Similarly, for a concept $v$ situated above a non-incident edge $v_1 v_2$, the conflict distance is defined as the absolute distance between its position $w$ and that of the upper edge concept $w_2$:
$$
d(w, f) = |w_2 - w|.
$$
Accordingly, we derive the forces using the unit vector between $w$ and $w_2$, where the direction of the force depends on the contribution of the element to the concept's extents or intents:
$$
\frac{\mathrm{d} |w_2 - w|}{\mathrm{d} n} = 
\begin{cases}
  \frac{w_2 - w}{|w_2 - w|} &, \text{if } e \in A_{v_2} \setminus A_v \text{ or } e \in B_{v_1} \setminus B_v \\
  \frac{w - w_2}{|w_2 - w|} &, \text{if } e \in A_v \setminus A_{v_2} \text{ or } e \in B_v \setminus B_{v_2} \\
  0 &, \text{otherwise\enspace.}
\end{cases}
$$
Similar to the dual case (a), the gradient pushes either the node or the non-incident edge apart based on the specific contribution of the element.

\bigskip
\noindent
(c) \; $(w_2 - w_1) \cdot (w - w_2) \leq 0$, $(w - w_1) \cdot (w- w_2) \geq 0$: \\
In the third case, where concept $v$ is situated between the edge concepts $v_1$ and $v_2$, the conflict distance is defined as the perpendicular distance from $w$ to the edge $f$:
$$
d(w, f) = \frac{A}{|f|} = \frac{|(w_1 - w) \times (w_2 - w)|}{|w_2 - w_1|}\enspace.
$$
Unlike the previous scenarios, this one requires a more sophisticated case distinction. The gradient strength is scaled based on the number of contributions of the element: it remains at full strength for elements that contribute to a single concept, but is scaled down when an element influences two concepts. The gradient is therefore derived based on the element's contribution to the concept's extent or intent:
$$
\frac{\mathrm{d} \frac{A}{|f|}}{\mathrm{d} e} =
\begin{cases}
  -\sqrt{\frac{(w_1 - w)^2 - |h|^2}{|f|^2}} \cdot \frac{n_+(f) \cdot l}{|f|} \cdot \MOD(e) &, \text{if } e \in A_{v_2} \setminus (A_v \cup A_{v_1}) \text{ or } e \in (B_v \cap B_{v_1}) \setminus B_{v_2} \\
  \sqrt{\frac{(w_2 - w)^2 - |h|^2}{|f|^2}} \cdot \frac{n_+(f) \cdot l}{|f|} \cdot \MOD(e) &, \text{if } e \in (A_v \cap A_{v_2}) \setminus A_{v_1} \text{ or } e \in B_{v_1} \setminus (B_v \cup B_{v_2}) \\
  -\frac{n_+(f) \cdot l}{|f|} &, \text{if } e \in A_{v_2} \setminus A_v \text{ or } e \in B_{v_1} \setminus B_v \\
  \frac{n_+(f) \cdot l}{|f|} &, \text{if } e \in A_v \setminus A_{v_2} \text{ or } e \in B_v \setminus B_{v_1}\enspace,
\end{cases}
$$
where 
$$
n_+(f) =
\begin{pmatrix} 
  -y(f) \\
  x(f)
\end{pmatrix}
$$
represents the edge vector $f$ rotated $\ang{90}$ in the positive direction to define the orthogonal vector. The factor $l$ acts as an orientation determined by the sign of the cross product $(w_1 - w) \times (w_2 - w)$, ensuring that the force always pushes the concept away from the non-incident edge. Accordingly, $l$ is set to $+1$ if $w$ lies to the left of the edge $f$ and $-1$ otherwise.

\subsubsection{Attractive Force}
\label{fdp_doubly_attractive_force}
For the attractive force, which pulls nodes together, we define the energy as the squared length of an edge:
$$
E_{att} = \sum_{v_1 v_2 \in \prec} |f|^2\enspace.
$$
To derive the acting forces, we simplify the calculation by excluding elements that contribute to both or neither of the concepts, as any adjustment to their positions would not improve the attractive energy. This leaves the objects present only in the extent of $v_2$ that pull the upper edge concept downward, and attributes present only in the intent of $v_1$ that pull the lower edge concept upward. Therefore, we derive the derivation as follows:
$$
F_{att}(n) = 2 \cdot \sum\limits_{\substack{v_1 v_2 \in \prec \\ e \in (A_2 \setminus A_1) \cup (B_1 \setminus B_2)}} (w_1 - w_2) \cdot \MOD(e)\enspace.
$$

\subsubsection{Gravitational Force}
\label{fdp_doubly_gravitational_force}
The gravitational force aims to keep the object vectors upward and attribute vectors downward-directed. Central to this model are the safe zones, where no forces act on the element vectors to prevent the model from becoming overly-constrained. To achieve a natural distribution of elements, the width of these zones is derived from the number of objects and attributes:
$$
\varphi_0(e) = 
\begin{cases}
  \frac{\pi}{|G| + 1} &, \text{if } e \in G \\
  \frac{\pi}{|M| + 1} &, \text{if } e \in M\enspace.
\end{cases}
$$
Based on these thresholds, we define the safe zone for objects in the upper semi-plane and dually for attributes in the lower semi-plane by the interval $[\varphi_0(e) \cdot \MOD(e), (\pi - \varphi_0(e)) \cdot \MOD(e)]$. 

The gravitational energy serves as a penalty for any element vector drifting outside its intended zone (see Figure~\ref{fig:grav_force}), thus approaching infinity as a vector approaches the horizontal axis. While the infinite forces for attribute vectors in the original approach enforced the final diagram to align with the conventions of the order relation, they introduce stability issues in doubly-additive line diagrams. If an object and an attribute vector simultaneously cross into the wrong semi-planes, their infinite forces cancel each other out, resulting in skewed layouts. Therefore, we define a linear penalty based on the vertical displacement, ensuring that the greater displacement dominates the resulting force. Accordingly, the resulting energy is defined as:
$$
E_{grav}(n) =
\begin{cases}
  \varphi(n) \cdot l + \cot(\varphi(n)) \sin^2 (\varphi_0(e)) + E(e) &, \text{if } \varphi_0(e) \cdot \MOD(e) \leq 0 \\
  y(n)^2 &, \text{if } \varphi_0(e) \cdot \MOD(e) > 0\enspace,
\end{cases}
$$
with a directional factor:
$$
l =
\begin{cases}
  +1 &, \text{if } 0 < \varphi(n) \cdot \MOD(e) < \varphi_o(e) \\
  -1 &, \text{if } \pi - \varphi_0(e) < \varphi(n) \cdot \MOD(e) < \pi \\
  \phantom{-}0 &, \text{ otherwise\enspace,}
\end{cases}
$$
which indicates whether a vector is too flat on the left ($-1$) or too flat on the right, ($+1$) and the integration constants:
$$
E(e) = 
\begin{cases}
-\varphi_0(e) \cdot \MOD(e) - \sin(\varphi_0(e)) \cos(\varphi_0(e)) \cdot \MOD(e) &, \text{if } l > 0 \\
-\varphi_0(e) \cdot \MOD(e) - \sin(\varphi_0(e)) \cos(\varphi_0(e)) \cdot \MOD(e) + \pi \cdot \MOD(e) &, \text{if } l < 0\enspace.
\end{cases}
$$
The resulting forces are derived as follows:
$$
F_{grav}(n) = 
\begin{cases}
  n_-(n) \cdot \frac{\sin^2 (\varphi(n)) - \sin^2 (\varphi_0(e))}{y(n)^2} \cdot l &, \text{if } \varphi_0(e) \cdot \MOD(e) \leq 0 \\
  -2 \cdot y(n) &, \text{if } \varphi_0(e) \cdot \MOD(e) > 0\enspace.
\end{cases}
$$

\begin{figure}
  \centering
  \input{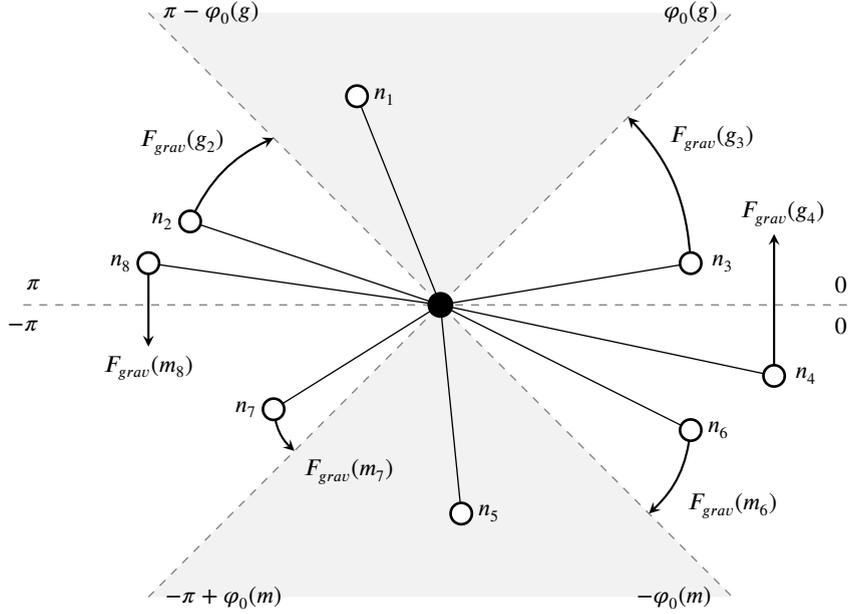}
  \caption{Visualization of the gravitational force acting on element vectors. Gray regions indicate the designated safe zones.}
  \label{fig:grav_force}
\end{figure}
\section{On the additivity of realizer-embedded diagrams}
\label{sec:additivity_dim_draw}
The planarity enhancer was developed to minimize the number of edge crossings in the initial layout, as a formal concept becomes locked within a cell due to the repulsive force acting more strongly if a concept approaches a non-incident edge. However, it becomes increasingly unstable as the number of concepts increases~\cite{ZschaligFDP}. Accordingly, since the performance of the FDP-optimizer heavily depends on the initial layout, we evaluate DimDraw as a structured alternative to compute the initial positions.

\subsection{Theoretical foundations}
\label{subsec:foundations_dim_draw}
Let $(P, \leq)$ be an ordered set. Two elements $p, q \in P$ are \emph{comparable} if either $p \leq q$ or $p > q$ holds; otherwise, they are \emph{incomparable}. The set of all incomparable pairs is denoed by $\INC(L)$. A total order $<_i$ is a \emph{linear extension} of $(P,\leq)$, if $p < q$ always implies $p <_i q$. A set $\RR =\{<_1, \dots, <_t\}$  of linear extensions is called a \emph{realizer} if its intersection restores the original relation, \ie, if $\bigcap_{i \in \{1, \dots, t\}} <_i = <$ holds. The \emph{order dimension} $\DIM(P)$ of the ordered set is defined as the cardinality of a minimal realizer~\cite{Dushnik1941, GanterFCA2024}.

For a two-dimensional concept lattice $\uB(\K)$ with a realizer $\RR = \{ <_1, <_2 \}$, DimDraw positions the node of a concept $c \in \BBK$ at $\POS(c) := (\POS(c)_{<_1}, \POS(c)_{<_2})$. This construction always ensures that the line diagram respects the order relation. Furthermore it guarantees that two nodes are never positioned at the same place. 
		
Because one always has $\POS(\bot)=(0,0)$ and $\POS(\top)=(|L|-1, |L|-1)$, these diagrams all seem squeezed and tilted to the right. For obtaining a more canonical diagram, one usually applies a linear coordinate transformation  with $(0, 1) \mapsto (-1, 1)$ and $(1, 0) \mapsto (1, 1)$, which turns the diagram by $45^{\circ}$ to the left, and then stretches the diagram horizontally by $\sqrt{2}$ and squeezes the diagram vertically by $\sqrt{\sfrac{1}{2}}$.\footnote{Figure~\ref{fig:not_attribute_additive} shows two rotated but non-stretched diagrams. The embedder-realized diagrams in the subsequent figures are all rotated and stretched. Rotating, stretching and squeezing are linear transformations and therefore do not change the status of a diagram of being additive or not.} This ensures that the top concept is positioned vertically above the bottom concept, provides a balanced width-to-height ratio, and positions all nodes on a grid. 

For $d$-dimensional lattices with $d \geq 3$, the same construction as above naturally yields a satisfactory $d$-dimensional line diagram. Experiments showed, however, that in general no satisfactory two-dimensional diagrams could be derived from them by searching for some optimal linear projection. Instead the lattice was first transformed into a two-dimensional ordered set by inserting a set of temporary relations $T \subseteq \INC(L)$ of minimal size. 

The augmented relation $\leq \! \cup \: T$ must remain an order on $L$ such that the resulting poset $P_T = (L, \leq \! \cup \: T)$ has a two-dimensional realizer. The concept positions are derived from the temporary realizer $\RR_T$ as described above, then the edges are drawn based on the original order relation $\leq$~\cite{Duerrschnabel2019, Duerrschnabel2023}.

We call these line diagrams \emph{realizer-embedded}, and denote the set of all realizer-embedded line diagrams of a lattice $L$ by $\reLD(L)$ (and its $d$-dimensional version by $\reLD^d(L))$. As these line diagrams automatically respect the lattice order, we have $\reLD(L)\subseteq\LD(L)$ and, in general, $\reLD^d(L)\subseteq\LD^d(L)$.

\subsection{Doubly-additive two-dimensional extensions}
\label{subsec:additive_two_dimensional_extension}
To combine DimDraw with the FDP-model as an approach to initialize the concept positions, the line diagram must be additive, as the forces are applied to the element vectors. Given that the original approach requires attribute-additive line diagrams, it is trivial to see that DimDraw layouts do not naturally fit this criterion, particularly if the lattice is non-distributive. Therefore, we investigate the $N_5$ and $M_3$, which are forbidden sublattices of any distributive lattice, according to Birkhoff's theorem. Figure~\ref{fig:not_attribute_additive} shows the symmetric DimDraw layouts, whereas in an attribute-additive line diagram, the bottom concept $c_\bot$ would be placed lower, as its position is strictly defined by the vector sum of the $\bigwedge$-irreducibles above.

\begin{figure}
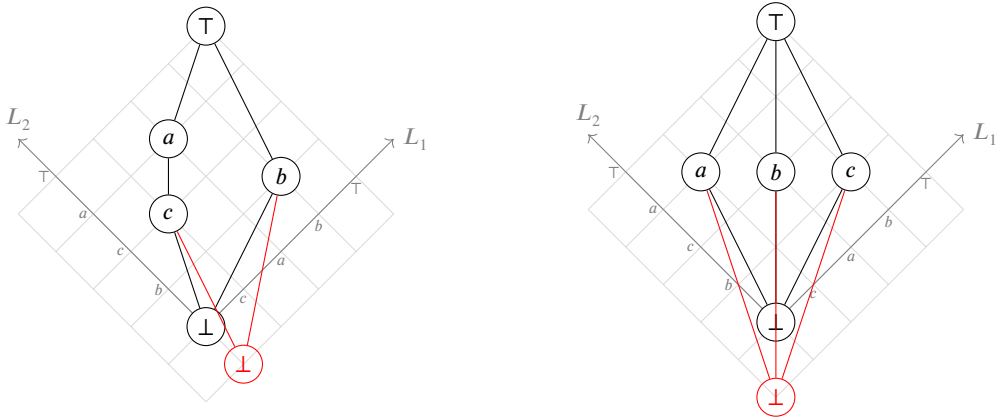

  \centering
  \begin{subfigure}[b]{0.45\textwidth}
    \centering
    \input{figs/dim_draw/additive_N5}
  \end{subfigure}
  \begin{subfigure}[b]{0.45\textwidth}
    \centering
    \input{figs/dim_draw/additive_M3}
  \end{subfigure}
  \caption{Realizer-embedded (black) diagrams for $N_5$ and $M_3$ overlaid with the attribute-additive diagram (red) derived from the attribute vectors.}
  \label{fig:not_attribute_additive}
\end{figure}

Ideally, every concept lattice would possess an additive DimDraw diagram. While experimental results are promising and show that DimDraw often produces the desired additivity, we need to know if this is a guaranteed property or merely a frequent byproduct of certain lattice structures. The following theorem states that there exist lattices which do no not have a line diagram (in 2D) that is both realizer-embedded and additive.
\begin{theorem}
\label{thm:additive_extension}
  Let $L$ be the free modular lattice $\FMT$ and $\K$ its standard context. $L$ does not have a line diagram that is both realizer-embedded and additive, \ie,
  $$
  \reLD(L) \cap \AddLD(\K) = \emptyset\enspace.
  $$
\end{theorem}

\begin{proof}[Proof of Theorem~\ref{thm:additive_extension}]
  For the standard context of $\FMT$ the SAT-solver of DimDraw~\cite{Duerrschnabel2023} identifies that the order relation has to be extended with at least $k = 15$ tuples to make it two-dimensional (see Figure~\ref{fig:additive_FM3} left). There are 12 different minimal solutions, but they are all symmetrically equivalent, differing only by rotation, mirroring, or the wing side. This behavior is a direct consequence of the automorphisms of $\UBBK$. Therefore, adding a temporary edge in the top cube determines the corresponding dual edge for the bottom cube to maintain the internal symmetry of the lattice. Since the underlying constraints are invariant under these symmetries, verifying a single representative diagram is sufficient to prove the properties for the entire class. Since the orthogonal projection into the additive space differs from the order diagram of DimDraw (see Figure~\ref{fig:additive_FM3} right), the realizer-embedded diagram (and any of its symmetric variants) is not additive.
\end{proof}

\begin{figure}
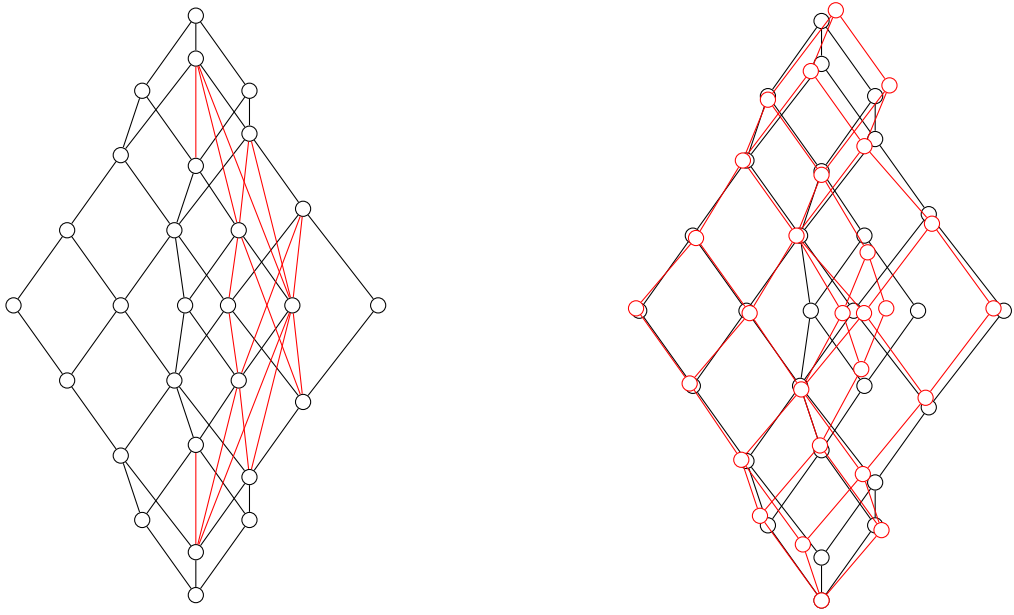

  \centering
  \begin{minipage}{0.5\textwidth}
    \centering
    \input{figs/dim_draw/tmp_FM3}
  \end{minipage}%
  \begin{minipage}{0.5\textwidth}
    \centering
    \input{figs/dim_draw/additive_FM3}
  \end{minipage}%
  \caption{DimDraw diagrams of the $\FMT$ (black). The left order diagram shows the inserted temporary edges (red) required for the two-dimensional extension. The right one is overlaid with the closest doubly-additive order diagram (red).}
  \label{fig:additive_FM3}
\end{figure}

\subsection{Doubly-additive realizers in the original order dimension}
\label{subsec:additive_realizer}
Since an additive two-dimensional extension does not exist for every concept lattice (as just shown above), one might expect that the realizer-embedded diagrams in the original order dimension of the lattice are additive (or at least one of them). If that were true then we might eventually apply a suitable linear projection to 2D to obtain a well readable diagram. 

Considering a finite lattice $L$ with $\dim(L)=d\geq 3$ and its standard context $\K:=\GMI:=(\JJ(L),\MM(L)\leq)$. To verify whether a given realizer $\RR = \{ <_1, \ldots, <_d \}$ actually satisfies the conditions for additivity, we have to find a mapping $\VEC \colon \JJ(L) \DC \MM(L) \to \R^d$ such that
$$
\POS(c)_{<_i} = \sum_{e \in A \cup (M \setminus B)} \VEC(e)_i\enspace, \textrm{ for all } c\in\uB(\K) \textrm{, for all } <_i \in \RR.
$$
The mapping $\VEC$ is thus defined by a system of $d\cdot |\B(\K)|$ linear equations over $d\cdot |\JJ(L)\dot\cup\MM(L)|$ variables (one variable for each dimension of $\VEC(e)$ for each $e\in \JJ(L)\dot\cup\MM(L)$). 

To avoid an exhaustive search through all possible realizers, we employ a \emph{Satisfiability Modulo Theories (SMT) solver}~\cite{Barrett2009}, which is an extension of classical SAT solvers to handle numerical constraints. By encoding our requirements into a formal model, the SMT solver will either find a doubly-additive realizer for a given concept lattice or prove that none exists.

We formally encode the requirements for a $d$-dimensional doubly-additive realizer into the SMT framework as follows: For each linear extension $<_i \in \RR$, we represent every element $e \in G \DC M$ as a real-valued variable $\VEC(e)_{<_i} \in \R$ and define the position of a formal concept $c \in \BBK$ as a natural number $\POS(c)_{<_i} \in \N_0$. To link these sets of SMT variables, we model the previously defined linear equations as constraints, ensuring that the sum of the element vectors equals the position of the corresponding formal concept in the linear extension. Furthermore, to enforce that the collection of linear extensions forms a valid realizer, we must model the underlying partial order. For all comparable pairs $(u, v) \in \; \leq$, the position of formal concept $u$ must be strictly less than the position of concept $v$ across all linear extensions:
$$
\forall (u, v) \in \; \leq, \forall <_i \in \RR \colon \POS(u)_{<_i} < \POS(v)_{<_i}\enspace.
$$
Conversely, to capture all incomparable pairs $(u, v) \in \INC(P)$, we must ensure that the order is not fixed. Accordingly, we require the relative order to be reversed in at least one linear extension:
$$
\forall \{ u, v \} \in \INC(P) \colon \exists <_1, <_2 \in \RR \text{ such that }  \POS(u)_{<_1} < \POS(v)_{<_1} \text{ and } \POS(v)_{<_2} < \POS(u)_{<_2}\enspace,
$$
which ensures that the intersection of all linear extensions recovers no order between incomparable pairs.

By fixing the position of the bottom concept $c_\bot$ and the top concept $c_\top$ for every linear extension $<_i \in \RR$:
$$
\begin{aligned}
  \forall <_i \in \RR \colon \POS(c_\bot)_{<_i} &= 0 \\
  \forall <_i \in \RR \colon \POS(c_\top)_{<_i} &= |\BBK| - 1\enspace,
\end{aligned}
$$
we ensure that all other concepts strictly lie between them. Because the positions are restricted to unique natural numbers, we enforce a uniform step-width of 1. Therefore, when the position assignments respect the constraints for all comparable and incomparable pairs, they form a valid realizer. Furthermore, the criterion of additivity ensures that these positions equal the sum of the objects in the extent and the attributes in the complement intent, which enforces a mapping of the order onto a consistent numerical coordinate system.

\begin{theorem}
\label{thm:additive_realizer_rd}
  Let $L$ be the free modular lattice $\FMT$ and $\K$ its standard context. $L$ does not have a $d$-dimensional line diagram that is both realizer-embedded and additive, \ie,
  $$
  \reLD^d(L) \cap \AddLD^d(\K) = \emptyset\enspace.
  $$
\end{theorem}

\begin{proof}[Proof of Theorem~\ref{thm:additive_realizer_rd}]
  From order theory, we know that the order dimension of $\FMT$ is $\DIM(\FMT) = 3$. Accordingly, to satisfy the theorem, we must show that no doubly-additive realizer of size $d = 3$ exists for the $\FMT$.

  Therefore, we encoded the constraints of a doubly-additive realizer $\RR = \{ <_1, <_2, <_3 \}$ into the SMT solver descibed above:
  \begin{enumerate}
    \item 
    $\forall <_i \in \RR, \forall e \in G \DC M \colon \VEC(e)_{<_i} \in \R$
    \item 
    $\forall <_i \in \RR, \forall c \in \BBK \colon \POS(c)_{<_i} \in \N_0 $
    \item 
    $\forall <_i \in \RR, \forall (A, B) \in \BBK \colon \sum_{e \in A \cup (M \setminus B)} \VEC(e)_{<_i}$
    \item 
    $\forall <_i \in \RR, \forall (u, v) \in \; \leq \colon \POS(u)_{<_i} < \POS(v)_{<_i}$
    \item 
    $\forall (u, v) \in \INC(\UBBK), \exists <_1, <_2 \in \RR \colon (\POS(u)_{<_1} < \POS(v)_{<_1}) \land (\POS(v)_{<_2} < \POS(u)_{<_2})$
    \item 
    $\forall <_i \in \RR \colon \POS(c_\bot)_{<_i} = 0$
    \item 
    $\forall <_i \in \RR \colon \POS(c_\top)_{<_i} = |\BBK| - 1\enspace$.
  \end{enumerate}
  The encoding was executed for the $\FMT$ using the SMT solver Z3~\cite{DeMoura2008}, which returned $\text{UNSAT}$, indicating that no such set of vectors and linear extensions exists that satisfies the doubly additivity constraints in $\R^d$ for the given lattice. 
\end{proof}

\section{DimFlux}
\label{sec:dim_flux}
Building on these methods, we introduce \textbf{DimFlux}, an algorithm for visualizing concept lattices that combines the advantages of order-theoretic structure and visual clarity. Firstly, we begin by computing an initial realizer-embedded line diagram, which is then refined by maximizing the conflict distance between concept nodes and non-incident edges. This optimization is specifically tuned to the analytical needs of Formal Concept Analysis, addressing common visualization taxonomies~\cite{Amar2005}, such as tracing paths, identifying up-sets and down-sets, or locating concepts linked to specific objects and attributes. Finally, by maintaining the realizer-embedded structure while ensuring concept nodes remain distinct from non-incident edges, DimFlux is designed to improve the interpretation of line diagrams of concept lattices.

\begin{figure}
    \centering

    \colorlet{inputcolor}{blue!10}
    \colorlet{actioncolor}{green!10}
    \colorlet{outputcolor}{red!10}

    \begin{tikzpicture}[
        base/.style={
            draw=black,
            thick,
            align=center,
            font=\small\sffamily,
            minimum height=1.0cm,
            inner xsep=5pt
        },
        action/.style={
            base,
            rectangle,
            rounded corners=6pt,
            minimum width=2.4cm,
            fill=actioncolor
        },
        object/.style={
            base,
            rectangle,
            minimum width=2.2cm,
            fill=white 
        },
        input/.style={object, fill=inputcolor},
        output/.style={object, fill=outputcolor},
        flow/.style={
            ->, 
            >={Stealth[scale=1.2]}, 
            thick,
            shorten >=2pt, 
            shorten <=2pt
        }
    ]

    \node (n1) [input] at (1.5,0) {Concept\\Lattice};
    \node (n2) [action] at (5.0,0) {Two-Dimensional\\Extension};
    \node (n3) [action] at (8.5,0) {Projection};
    \node (n4) [action] at (12.0,0) {Force-Directed\\Refinement};
    \node (n5) [output] at (15.5,0) {DimFlux\\Diagram};

    \begin{scope}[flow]
        \draw (n1.east) -- (n2.west);
        \draw (n2.east) -- (n3.west);
        \draw (n3.east) -- (n4.west);
        \draw (n4.east) -- (n5.west);
    \end{scope}

    \begin{scope}[on background layer]
        \node[fill=inputcolor!50, rounded corners, fit=(n1), label={above:{\textbf{Input}}}] {};
        \node[fill=actioncolor!50, rounded corners, fit=(n2) (n3) (n4), label={above:{\textbf{Algorithm Steps}}}] {};
        \node[fill=outputcolor!50, rounded corners, fit=(n5), label={above:{\textbf{Output}}}] {};
    \end{scope}

    \end{tikzpicture}
    \caption{The logical flow of our DimFlux algorithm.}
    \label{fig:dim_flux}
\end{figure}
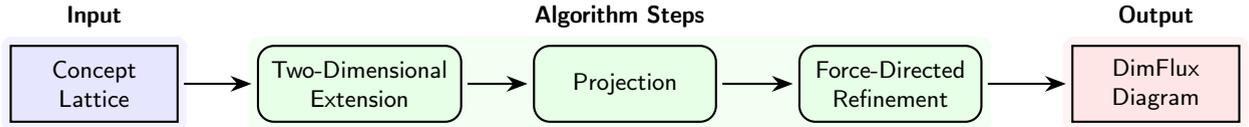

\subsection{Workflow}
\label{subsec:workflow}
Given a formal context $\K = \GMI$, our algorithm generates a line diagram of the corresponding concept lattice through a three-step process. The DimDraw algorithm provides the structural foundation by translating order-theoretic properties into geometric coordinates. For a given set of formal concepts, it initially maps each concept node to a point in $\R^2$ based on a two-dimensional extension of the lattice. To transform this initial layout into a doubly-additive line diagram, these coordinates are projected into the additive space (see Section~\ref{sec:additive_projection}), using the Gram-Schmidt orthogonalization method to calculate the orthogonal projection into the space of additive placements $\AddLD(\K)$. Finally, we refine the line diagram through the extended force-directed placement model (see Section~\ref{sec:fdp_doubly}), which focuses on maximizing the conflict distance between concept nodes and non-incident edges, effectively improving the readability of the line diagram.

While concept lattices of real-world data may be significantly larger, only lattices with up to 20-30 elements are drawn in practice. For larger concept lattices, decomposition methods like subdirect decomposition are applied to split the lattice into smaller parts, enabling alternative visualizations such as nested line diagrams.

The DimFlux implementation (see Figure~\ref{fig:dim_flux}), including the two-dimensional extension, projection, and force-directed placement model, is available as open-source software.\footnote{The source code is available at \url{https://github.com/marcelnoehre/dim-flux}} Furthermore, the first two steps of the workflow can be replaced by computing an initial layout using the doubly-additive version of the planarity enhancer (as described in Subsection~\ref{subsec:fdp_doubly_initialization}) to align with the original approach by Zschalig.

\subsection{Runtime Discussion}
\label{subsec:runtime}
The total runtime complexity of our algorithm can be traced back to the three main stages: the initial layout generation using DimDraw, the projection into the additive space, and the force-directed refinement. The initial layout generation, which relies on the Rust implementation of DimDraw~\cite{Duerrschnabel2025}, runs in $O(n^3)$ time due to the transitive orientation algorithm by Golumbic~\cite{Golumbic1977}. For each following iteration that extracts bipartite subgraphs to construct the transitive incompatibility graph, the complexity is $O(n^2)$, where $n$ is the number of formal concepts; however, the solver terminates after a small number of iterations in practice~\cite{Duerrschnabel2023}. The subsequent projection step, based on the Gram-Schmidt orthogonalization, scales at $O(nm^2)$~\cite{Lyubashevsky2015}, where $n$ is the number of formal concepts and $m$ is the number of elements $|G| + |M|$. Finally, the force-directed placement, relying on a Conjugate Gradient optimizer, has a complexity of $O(m\sqrt{\kappa})$~\cite{Shewchuk1994}, where $\kappa$ is the number of iterations required for convergence. In summary, the projection and force-based refinement, which also terminate after a small number of iterations, are computationally less expensive than the initial layout generation, which remains the primary runtime bottleneck. Consequently, the algorithm is well-suited for medium-sized concept lattices where the one-time $O(n^3)$ cost remains manageable. Moreover, this layout cost is well below the complexity of constructing the lattice itself, which can scale exponentially at $O(2^m)$ for $m = \operatorname{min}(|G|, |M|)$.

\section{Evaluation}
\label{sec:evaluation}
To evaluate our approach, we generated line diagrams for all 126 lattices with four meet-irreducibles (which equal all 126 concept lattices from reduced formal contexts with four attributes), together with some real-world examples from the book~\cite{GanterFCA2024}. A comprehensive collection of our results is available on Zenodo~\cite{NoehreDimFlux}, including the lattice drawings, their distance to the hand-drawn versions, and the energy scores from the doubly-additive force model. For non-doubly-additive diagrams, we used the projected version (indicated by an $^{\ast}$). Lattices are identified either by their index within the set of all 126 lattices with four meet-irreducibles or by name for the real-world examples. Diagrams are given both as coordinates in lectic order and as drawings in a PDF document. For each lattice, we generated the following line diagrams:

\begin{enumerate}
	\item \textbf{Hand-Drawn:} 
   Diagrams that were manually created by an expert for all lattices with four meet-irreducibles and for the real-world examples.
	\item \textbf{Attribute-Additive FDP:} 
  Line diagrams computed by the original force-directed placement model for attribute-additive line diagrams~\cite{ZschaligFDP}.
  \item \textbf{Doubly-Additive FDP:} 
  Our extension of the force-directed placement model for doubly-additive line diagrams.
  \item \textbf{DimDraw:}
  Order diagrams through a two-dimensional extension~\cite{Duerrschnabel2019}.
  \item \textbf{DimFlux:}
  Line diagrams derived from a two-dimensional extension that are projected into the additive space and subsequently refined using a force-directed placement model.
\end{enumerate}

While universal quality metrics for graph drawings exist~\cite{Mooney2025}, they failed to reveal clear trends across the different drawing algorithms in our study (see~\cite[Section~3]{NoehreDimFlux}). Notably, hand-drawn versions often scored poorly on these standard metrics, suggesting that the visual criteria for line diagrams of lattices deviate from those of general graph drawing.

Consequently, we measure the Euclidean distance between the algorithmic results and their hand-drawn counterparts as our primary metric. By representing each diagram of a lattice $L$ as a vector in $\mathbb{R}^{2|L|}$, we can quantify how much the calculated node positions deviate from the ``optimal'' hand-drawn layout. To ensure comparability across different lattices, all diagrams were normalized to a standard height of ten units.

\begin{figure}
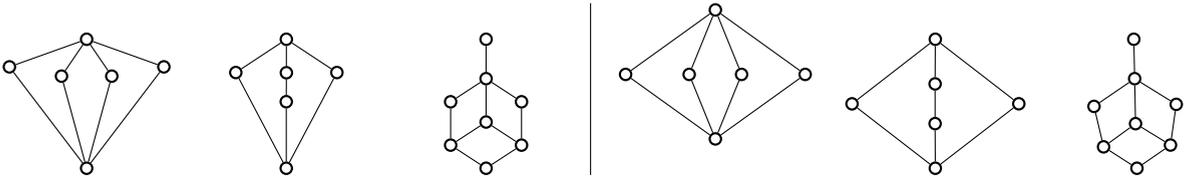

  \centering
  \begin{subfigure}[b]{0.16\textwidth}
    \centering
    \input{figs/evaluation/sup_inf_attribute_22}
  \end{subfigure}%
  \begin{subfigure}[b]{0.16\textwidth}
    \centering
    \input{figs/evaluation/sup_inf_attribute_19}
  \end{subfigure}%
  \begin{subfigure}[b]{0.16\textwidth}
    \centering
    \input{figs/evaluation/sup_inf_attribute_8}
  \end{subfigure}%
  \vrule
  \hfill
  \begin{subfigure}[b]{0.16\textwidth}
    \centering
    \input{figs/evaluation/sup_inf_doubly_22}
  \end{subfigure}%
  \hfill
  \begin{subfigure}[b]{0.16\textwidth}
    \centering
    \input{figs/evaluation/sup_inf_doubly_19}
  \end{subfigure}%
  \begin{subfigure}[b]{0.16\textwidth}
    \centering
    \input{figs/evaluation/sup_inf_doubly_8}
  \end{subfigure}%
  \caption{Comparison of attribute-additive (left) and doubly-additive (right) line diagrams of concept lattices 22, 19, and 8 (left to right), each initialized using its respective Sup-Inf graph.}
  \label{fig:evaluation_fdp_layer}
\end{figure}

\begin{figure}
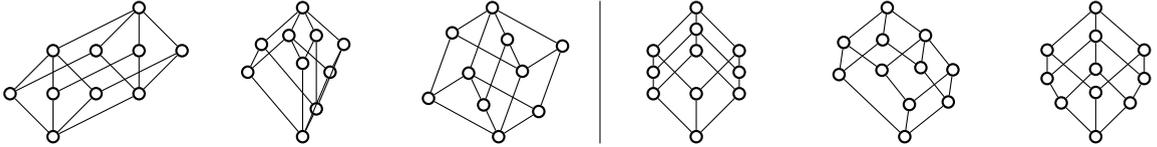

  \centering
  \begin{subfigure}[b]{0.16\textwidth}
    \centering
    \input{figs/evaluation/hand_drawn_77}
  \end{subfigure}%
  \begin{subfigure}[b]{0.16\textwidth}
    \centering
    \input{figs/evaluation/sup_inf_attribute_77}
  \end{subfigure}%
  \begin{subfigure}[b]{0.16\textwidth}
    \centering
    \input{figs/evaluation/sup_inf_doubly_77}
  \end{subfigure}%
  \vrule
  \begin{subfigure}[b]{0.16\textwidth}
    \centering
    \input{figs/evaluation/hand_drawn_98}
  \end{subfigure}%
  \begin{subfigure}[b]{0.16\textwidth}
    \centering
    \input{figs/evaluation/sup_inf_attribute_98}
  \end{subfigure}%
  \begin{subfigure}[b]{0.16\textwidth}
    \centering
    \input{figs/evaluation/sup_inf_doubly_98}
  \end{subfigure}%
  \caption{Hand-drawn (left), attribute-additive (center), and doubly-additive (right) line diagrams for concept lattices 77 and 98, initialized using their respective Sup-Inf graphs.}
  \label{fig:evaluation_fdp_initialization}
\end{figure}

\subsection{Doubly-Additive FDP}
\label{subsec:evaluation_doubly}
Firstly, we compare the original attribute-additive approach by Zschalig against the adapted version for doubly-additive line diagrams. While additive line diagrams are generally effective at highlighting symmetries within concept lattices, the doubly-additive representation is inherently more balanced because it integrates both object and attribute influences. This structural advantage is particularly evident in non-distributive lattices that contain sublattices isomorphic to $N_5$ or $M_3$ (see figure~\ref{fig:evaluation_fdp_layer}).

Another benefit of the extended forces is the improvement in structural layering. In the original attribute-additive approach, vectors are often assigned a similar height, which results in vertical offsets, particularly within non-distributive lattices. Considering $M_4$ (lattice 22, see Figure~\ref{fig:evaluation_fdp_layer}), the sum of the contributing attribute vectors pushes the lower concept node further down, which also skews the vertical layout of surrounding structures. The advantage of the adapted FDP-model for doubly-additive line diagrams is directly visible in the $N_5$ sublattices (lattice 19), as the nodes of the concepts in the shorter chains are not vertically aligned between the two incomparable elements from the longer chain. In the doubly-additive version, the diagrams maintain a more intuitive and balanced vertical distribution.

However, the extended force-directed placement model does not perform better for all considered concept lattices. Due to the higher degree of freedom introduced by the larger vector set, the concept nodes drift away, as the underlying forces are designed to maximize the conflict distance. Particularly, double-irreducible nodes are affected by this behavior, which is frequently observed in the top nodes of the $S_7$ sublattice (lattice 8,  see Figure~\ref{fig:evaluation_fdp_layer}). However, this issue becomes less pronounced in larger concept lattices, due to surrounding non-incident edges that push the concept nodes back into the $S_7$ structure.

In summary, the doubly-additive force-directed placement version yields more vertically balanced layouts with improved spacing due to the higher degree of freedom. The doubly-additive version reveals common sublattices across all lattices --- such as $N_5$, $M_3$, $S_7$, or $B_3$ --- in a stable shape that is close to hand-drawn versions. By revealing these well-known patterns with such structural clarity, the algorithm enables users to comprehend the underlying formal context more effectively, enabling the desired taxonomic tasks~\cite{Amar2005} such as identifying concepts linked to specific objects and attributes through visual recognition.

\begin{figure}
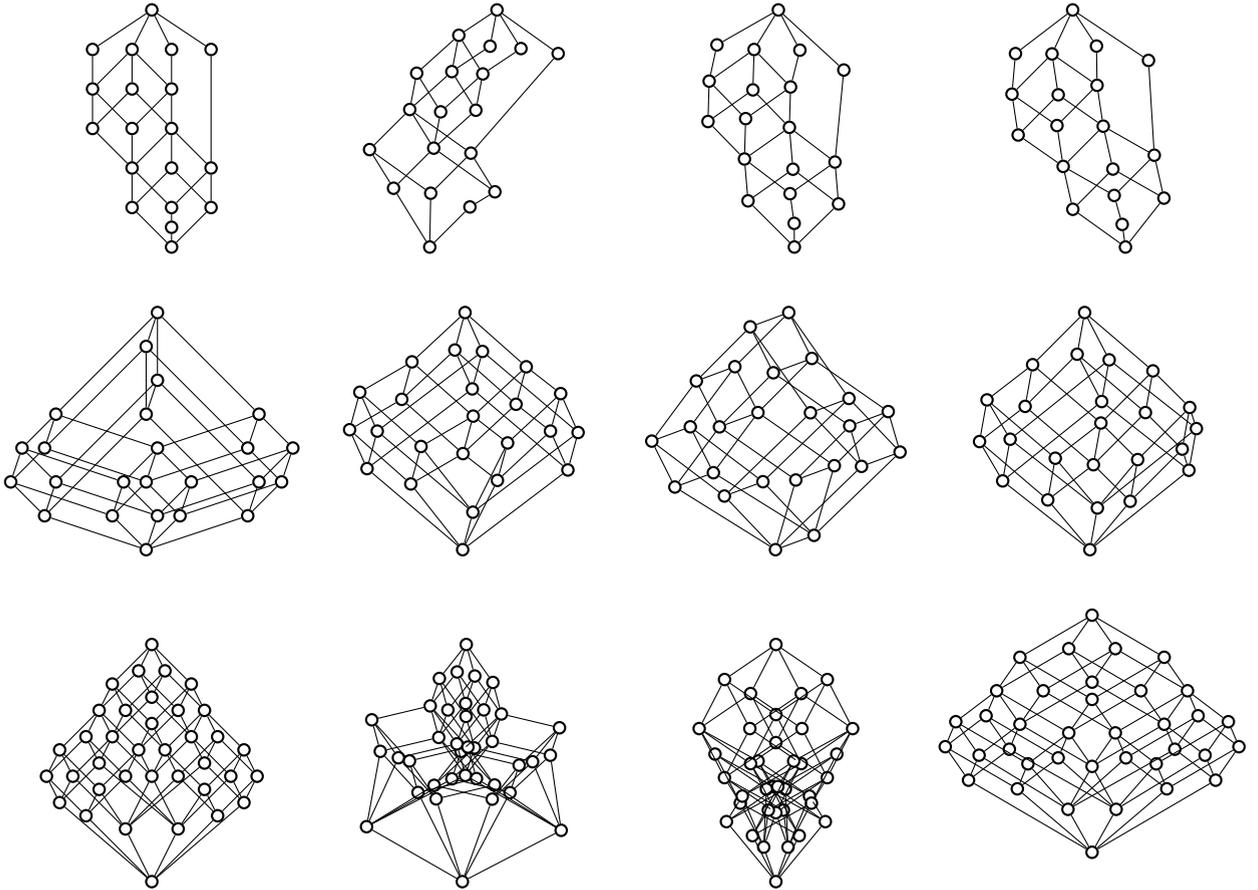

  \centering
  \begin{subfigure}[b]{0.33\textwidth}
    \centering
    \input{figs/evaluation/hand_drawn_forum_romanum}
  \end{subfigure}%
  \begin{subfigure}[b]{0.33\textwidth}
    \centering
    \input{figs/evaluation/sup_inf_attribute_forum_romanum}
  \end{subfigure}%
  \begin{subfigure}[b]{0.33\textwidth}
    \centering
    \input{figs/evaluation/sup_inf_doubly_forum_romanum}
  \end{subfigure}%

  \vspace{2em}

  \begin{subfigure}[b]{0.33\textwidth}
    \centering
    \input{figs/evaluation/hand_drawn_driving_concepts}
  \end{subfigure}%
  \begin{subfigure}[b]{0.33\textwidth}
    \centering
    \input{figs/evaluation/sup_inf_attribute_driving_concepts}
  \end{subfigure}%
  \begin{subfigure}[b]{0.33\textwidth}
    \centering
    \input{figs/evaluation/sup_inf_doubly_driving_concepts}
  \end{subfigure}%

  \vspace{2em}

  \begin{subfigure}[b]{0.33\textwidth}
    \centering
    \input{figs/evaluation/hand_drawn_convex_ordinal}
  \end{subfigure}%
  \begin{subfigure}[b]{0.33\textwidth}
    \centering
    \input{figs/evaluation/sup_inf_attribute_convex_ordinal}
  \end{subfigure}%
  \begin{subfigure}[b]{0.33\textwidth}
    \centering
    \input{figs/evaluation/sup_inf_doubly_convex_ordinal}
  \end{subfigure}%
  \caption{Hand-drawn (left), attribute-additive (center) and doubly-additive (right) diagrams, initialized using the respective Sup-Inf graphs, for the concept lattices ``Forum Romanum'' (top row), ``Driving Concepts'' (center row) and ``Convex-ordinal scale'' (bottom row).}
  \label{fig:evaluation_fdp_real_world}
\end{figure}

\subsection{Extended Planarity Enhancer}
\label{subsec:extended_planarity_enhancer}
The most significant improvement in the extension of the FDP approach to doubly-additive diagrams lies in the initialization of the line diagrams. By operating on more structural information than the original approach, combining all objects and attributes into a linear order, the extended planarity enhancer demonstrates improvements for medium-sized concept lattices, such as lattice 77 (see Figure~\ref{fig:evaluation_fdp_initialization}). Compared to the original layout, in which a concept at the bottom left prevents the forces from revealing the two connected $M_3$ sublattices, the doubly-additive extension successfully reveals the underlying structure. While this adjustment can lead to a higher number of edge crossings, which conflicts with the original goal of this initialization approach, the resulting layouts often better reflect the structure of the data, as seen in lattice 98 (see Figure~\ref{fig:evaluation_fdp_initialization}), where it reveals the inherent symmetry.

When applied to larger, real-world datasets, our results align with Zschalig's observations regarding the limitations of the planarity enhancer: the approach is effective for small lattices but becomes increasingly unstable as the number of concepts increases. For the Forum Romanum concept lattice with a relatively simple sub-structure of two stacked $B_3$ sublattices, both initializations based on the respective Sup-Inf graphs perform adequately, while the extension improves by centering the longer chain at the bottom. However, this behavior becomes even more unstable with increasing structural complexity, as evident in the driving concepts lattice.

\begin{figure}
  \centering
  \begin{subfigure}[b]{0.33\textwidth}
    \centering
    \input{figs/evaluation/hand_drawn_forum_romanum}
  \end{subfigure}%
  \begin{subfigure}[b]{0.33\textwidth}
    \centering
    \input{figs/evaluation/dim_draw_forum_romanum}
  \end{subfigure}%
  \begin{subfigure}[b]{0.33\textwidth}
    \centering
    \input{figs/evaluation/dim_flux_forum_romanum}
  \end{subfigure}%

  \vspace{2em}

  \begin{subfigure}[b]{0.33\textwidth}
    \centering
    \input{figs/evaluation/hand_drawn_driving_concepts}
  \end{subfigure}%
  \begin{subfigure}[b]{0.33\textwidth}
    \centering
    \input{figs/evaluation/dim_draw_driving_concepts}
  \end{subfigure}%
  \begin{subfigure}[b]{0.33\textwidth}
    \centering
    \input{figs/evaluation/dim_flux_driving_concepts}
  \end{subfigure}%

  \vspace{2em}

  \begin{subfigure}[b]{0.33\textwidth}
    \centering
    \input{figs/evaluation/hand_drawn_convex_ordinal}
  \end{subfigure}%
  \begin{subfigure}[b]{0.33\textwidth}
    \centering
    \input{figs/evaluation/dim_draw_convex_ordinal}
  \end{subfigure}%
  \begin{subfigure}[b]{0.33\textwidth}
    \centering
    \input{figs/evaluation/dim_flux_convex_ordinal}
  \end{subfigure}%
  \caption{Hand-drawn (left), DimDraw (center) and DimFlux (right) diagrams, for the concept lattices ``Forum Romanum'' (top row), ``Driving Concepts'' (center row) and ``Convex-ordinal scale'' (bottom row).}
  \label{fig:evaluation_dim_flux_real_world}
\end{figure}

\subsection{DimFlux}
\label{subsec:evaluation_dimflux}
For small lattices, DimDraw produces layouts comparable to the planarity enhancer. The structural advantages of realizer-embedded line diagrams become more evident in larger, real-world examples, where structural stability is preserved. However, the grid-based positioning, inherent to the underlying two-dimensional extension, is less refined in terms of local readability, specifically regarding the conflict distance between nodes and non-incident edges.

Therefore, we propose computing the initial layout using DimDraw rather than the planarity enhancer, which necessitates projecting the diagram into the space of valid additive diagrams $\AddLD(\K)$. For many small lattices and several larger examples, the original DimDraw diagram is already additive, which is shown in our supplementary collection~\cite{NoehreDimFlux}, where only the energy scores marked with an asterisk ($^{\ast}$) required no projection for DimDraw diagrams. In cases where DimDraw diagrams are non-additive, the distortion required to achieve additivity is minimal, which allows us to preserve the structured layout of realizer-embedded diagrams for subsequent readability-driven refinements.

Since the repulsive force prevents nodes from leaving their initial cells, the global layout remains stable, and the introduction of new edge crossings is effectively prevented. Consequently, the subsequent force-directed placement is dedicated entirely to enhancing the drawing's readability (see Figure~\ref{fig:evaluation_dim_flux_real_world}). Accordingly, this approach successfully combines the structured global layout of realizer-embedded line diagrams with improved local clarity, achieved by maximizing the conflict distance between concept nodes and non-incident edges.

Comparing the resulting drawings to hand-drawn versions, using our metric that treats each lattice drawing as a vector in $\R^{2|L|}$, reveals no single global trend, but DimFlux achieves higher scores (closer proximity to hand-drawn versions) primarily when it produces a more compact layout, as evidenced by the ``Forum Romanum'' and ``convex-ordinal scale'' concept lattices. However, the primary contribution remains the previously mentioned improvement in readability, which yields recognizable sublattices in well-known shapes. In the examples shown in Figure~\ref{fig:evaluation_dim_flux_real_world}, this behavior is most pronounced for the $B_3$ sublattices. In summary, DimFlux generates line diagrams that preserve a structured global realizer-embedding while optimizing local readability, thereby improving taxonomic tasks such as path tracing and the localization of specific concepts.

\section{Limitations}
\label{sec:limit}
We are pleased to see that the DimFlux algorithm, which combines several proven techniques, often produces easy-to-read line diagrams. However, we do not want to promise too much: even DimFlux does not always find a good diagram, nor does it always find the best possible one. As size increases, many concept lattices become too complex to be presented informatively in a single diagram. \FCA then offers alternative representations such as \emph{nested line diagrams}~\cite{GanterFCA2024}. Sometimes geometric representations are preferable: e.g., for $PG(2,2)$, the lattice of the subvector spaces of the three-dimensional vector space over the binary field, the geometric picture as a plane with points and lines is much more concise than a line diagram. Another lattice which is “difficult to draw” is that of all equivalence relations on a four-element set, see Figure~\ref{fig:part4}.

What constitutes ``the best'' diagram may also depend on the interpretation objective. As an example (see Figure~\ref{fig:jektivitaeten}), we show a lattice that reflects the result of a conceptual exploration. The study examined how the properties “injective” and “surjective” can be distributed across mappings $f$, $g$ and their compositions $f\circ g$ and $g\circ f$.

These two examples demonstrate that while DimFlux can often produce well-structured and readable line diagrams, it cannot always handle the inherent complexity of larger lattices or the need for specialized visualizations, such as nested line diagrams. Furthermore, there is no single standard for ``the best'' line diagram; although universal quality metrics for graph drawing exist~\cite{Mooney2025}, they are not generally applicable for lattice drawings.

\begin{figure}
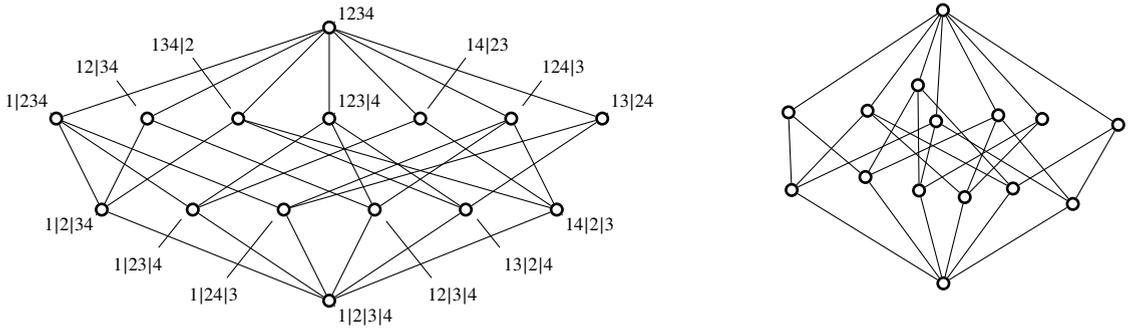

    \centering
    \begin{minipage}{0.6\textwidth}
        \centering
        \input{figs/limitations/hand_drawn_part4}
    \end{minipage}%
    \begin{minipage}{0.4\textwidth}
        \centering
        \input{figs/limitations/dim_flux_part4}
    \end{minipage}%
    \caption{Line diagrams of the lattice of all equivalence relations on a four-element set, hand-drawn~\cite{Biehl2017} (left) and DimFlux (right).}
    \label{fig:part4}
\end{figure}

\begin{figure}
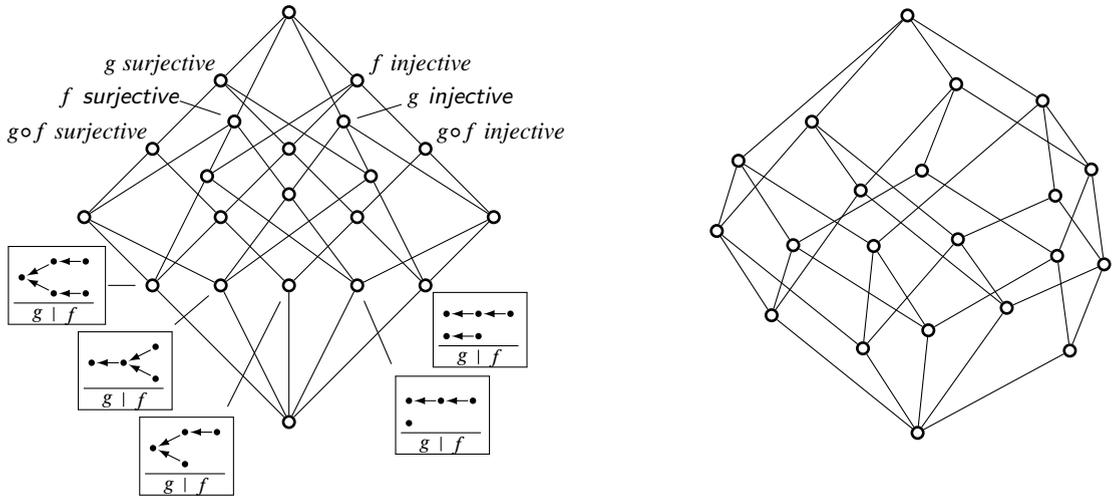

    \centering
    \begin{minipage}{0.6\textwidth}
        \centering
        \input{figs/limitations/hand_drawn_jektivitaeten}
    \end{minipage}%
    \begin{minipage}{0.4\textwidth}
        \centering
        \input{figs/limitations/dim_flux_jektivitaeten}
    \end{minipage}%
  \caption{The line diagram generated by DimFlux (right) does not reflect the symmetry that is visible in the hand-drawn version in~\cite{Ganter.2013} (left).}
  \label{fig:jektivitaeten}
\end{figure}
\section{Conclusion}
\label{sec:conclusion}
Extending Zschalig's force-directed placement approach to doubly-additive line diagrams yields the expected improvements regarding symmetry and layering. While the adapted planarity enhancer offers noticeable improvements for small and medium-sized lattices, it remains unstable for larger lattices, which is a known issue from the original work.

These issues are addressed by \textbf{DimFlux}, which utilizes realizer-embedded line diagrams as an initial layout to ensure stability, particularly for larger concept lattices. By projecting these diagrams into the space of additive diagrams, we enable the application of the doubly-additive force-directed placement, which prevents nodes from leaving their initial cells while maximizing the conflict distance between concept nodes and non-incident edges. As a result, our approach enhances the local readability while strictly preserving the global structure. In particular, DimFlux reveals common sublattices in well-known shapes, thus enabling more efficient taxonomic tasks such as path tracing and the localization of specific concepts.

A remaining challenge is the lack of an intrinsic slope awareness in the force-directed model, which would further increase the readability of the line diagram. While the DimFlux diagrams are readable for a human reader, chains are rarely placed on straight lines, which can detract from the overall visual clarity. Therefore, future work should address the issue of slope awareness, which has already been tackled in force-directed algorithms such as ReDraw~\cite{Duerrschnabel2021}. Accordingly, one could include slope-alignment forces into the force-directed model or implement a post-processing step to straighten chains without breaking the global optimization achieved by the force-directed refinement. One such approach is explained in Subsection~\ref{subsubsec:snap}.

\bibliographystyle{cas-model2-names}
\bibliography{bibliography}

@ARTICLE{sugiyama1981methods,
	author={Sugiyama, Kozo and Tagawa, Shojiro and Toda, Mitsuhiko},
	journal={IEEE Transactions on Systems, Man, and Cybernetics}, 
	title={Methods for Visual Understanding of Hierarchical System Structures}, 
	year={1981},
	volume={11},
	number={2},
	pages={109-125},
	doi={10.1109/TSMC.1981.4308636}
}

@InProceedings{freese2004automated,
	author="Freese, Ralph",
	editor="Eklund, Peter",
	title="Automated Lattice Drawing",
	booktitle="Concept Lattices",
	year="2004",
	publisher="Springer Berlin Heidelberg",
	address="Berlin, Heidelberg",
	pages="112--127",
	isbn="978-3-540-24651-0"
}

@article{aeschlimann1992drawing,
	author = {Aeschlimann, A. and Schmid, J.},
	day = 01,
	description = {Drawing orders using less ink | Order | Springer Nature Link},
	doi = {10.1007/BF00419035},
	issn = {1572-9273},
	journal = {Order},
	month = mar,
	number = 1,
	pages = {5--13},
	title = {Drawing orders using less ink},
	url = {https://doi.org/10.1007/BF00419035},
	volume = 9,
	year = 1992
}

@book{GanterFCA2024,
	title         = {Formal Concept Analysis - Mathematical Foundations},
	subtitle      = {Mathematical Foundations},
	author        = {Bernhard Ganter and Rudolf Wille},
	publisher     = {Springer Cham},
	year          = {2024},
	edition       = {2},
	pages         = {XII, 370},
	isbn          = {978-3-031-63421-5},
	ebookisbn     = {978-3-031-63422-2},
	softcoverisbn = {978-3-031-63424-6},
	address       = {Cham, Switzerland}
}

@book{Ganter.2013,
	title={Diskrete {M}athematik: Geordnete {M}engen},
	author={Ganter, Bernhard},
	year={2013},
	publisher={Springer}
}

@incollection{abdulla2025rises,
	address = {Cham, Switzerland},
	author = {Abdulla, Mohammad and Hille, Tobias and Dürrschnabel, Dominik and Stumme, Gerd},
	booktitle = {Conceptual Knowledge Structures},
	editor = {Cellier, Peggy and Ganter, Bernhard and Missaoui, Rokia},
	journal = {SpringerLink},
	month = {09},
	pages = {392--407},
	publisher = {Springer},
	series = {Lecture Notes in Computer Science},
	title = {Rises for Measuring Local Distributivity in Lattices},
	type = {Publication},
	volume = 15941,
	year = 2025
}

@article{Dushnik1941,
	ISSN = {00029327, 10806377},
	URL = {http://www.jstor.org/stable/2371374},
	author = {Ben Dushnik and E. W. Miller},
	journal = {American Journal of Mathematics},
	number = {3},
	pages = {600--610},
	publisher = {Johns Hopkins University Press},
	title = {Partially Ordered Sets},
	urldate = {2026-03-12},
	volume = {63},
	year = {1941}
}

@article{GanterConflict2004,
	author       = {Ganter, Bernhard},
	title        = {Conflict Avoidance in Additive Order Diagrams},
	journal      = {JUCS - Journal of Universal Computer Science},
	year         = 2004,
	volume       = 10,
	number       = {(8)},
	month        = aug,
	doi          = {10.3217/jucs-010-08-0955},
	url          = {https://doi.org/10.3217/jucs-010-08-0955},
}

@Inbook{wille1989lattices,
	author="Wille, Rudolf",
	editor="Rival, Ivan",
	title="Lattices in Data Analysis: How to Draw Them with a Computer",
	bookTitle="Algorithms and Order",
	year="1989",
	publisher="Springer Netherlands",
	address="Dordrecht",
	pages="33--58",
	isbn="978-94-009-2639-4",
	doi="10.1007/978-94-009-2639-4_2",
	url="https://doi.org/10.1007/978-94-009-2639-4_2"
}

@article{eades84,
	author = {Eades, P.},
	booktitle = {Congressus Numerantium, 42, 149-160.},
	journal = {Congressus Numerantium},
	pages = {149-160},
	title = {A heuristic for graph drawing},
	volume = 42,
	year = 1984
}

@article{Kamada1989,
	title = {An algorithm for drawing general undirected graphs},
	journal = {Information Processing Letters},
	volume = {31},
	number = {1},
	pages = {7-15},
	year = {1989},
	issn = {0020-0190},
	doi = {https://doi.org/10.1016/0020-0190(89)90102-6},
	url = {https://www.sciencedirect.com/science/article/pii/0020019089901026},
	author = {Tomihisa Kamada and Satoru Kawai}
}

@article{Fruchterman1991,
	author = {Fruchterman, Thomas M. J. and Reingold, Edward M.},
	title = {Graph drawing by force-directed placement},
	journal = {Software: Practice and Experience},
	volume = {21},
	number = {11},
	pages = {1129-1164},
	doi = {https://doi.org/10.1002/spe.4380211102},
	url = {https://onlinelibrary.wiley.com/doi/abs/10.1002/spe.4380211102},
	eprint = {https://onlinelibrary.wiley.com/doi/pdf/10.1002/spe.4380211102},
	year = {1991}
}

@article{Hestenes1952,
	title={Methods of conjugate gradients for solving linear systems},
	author={Magnus R. Hestenes and Eduard Stiefel},
	journal={Journal of research of the National Bureau of Standards},
	year={1952},
	volume={49},
	pages={409-435},
	url={https://api.semanticscholar.org/CorpusID:2207234}
}

@Article{Golumbic1977,
	author={Golumbic, M. C.},
	title={The complexity of comparability graph recognition and coloring},
	journal={Computing},
	year={1977},
	month={Sep},
	day={01},
	volume={18},
	number={3},
	pages={199-208},
	issn={1436-5057},
	doi={10.1007/BF02253207},
	url={https://doi.org/10.1007/BF02253207}
}

@techreport{Shewchuk1994,
	author = {Shewchuk, Jonathan R},
	title = {An Introduction to the Conjugate Gradient Method Without the Agonizing Pain},
	year = {1994},
	publisher = {Carnegie Mellon University},
	address = {USA}
}

@InProceedings{Lyubashevsky2015,
	author="Lyubashevsky, Vadim
	and Prest, Thomas",
	editor="Oswald, Elisabeth
	and Fischlin, Marc",
	title="Quadratic Time, Linear Space Algorithms for Gram-Schmidt Orthogonalization and Gaussian Sampling in Structured Lattices",
	booktitle="Advances in Cryptology -- EUROCRYPT 2015",
	year="2015",
	publisher="Springer Berlin Heidelberg",
	address="Berlin, Heidelberg",
	pages="789--815",
	isbn="978-3-662-46800-5"
}

@InProceedings{Duerrschnabel2021,
	author="D{\"u}rrschnabel, Dominik
	and Stumme, Gerd",
	editor="Braud, Agn{\`e}s
	and Buzmakov, Aleksey
	and Hanika, Tom
	and Le Ber, Florence",
	title="Force-Directed Layout of Order Diagrams Using Dimensional Reduction",
	booktitle="Formal Concept Analysis",
	year="2021",
	publisher="Springer International Publishing",
	address="Cham",
	pages="224--240",
	isbn="978-3-030-77867-5"
}

@InProceedings{Zschalig2005,
	author="Zschalig, Christian",
	editor="Ganter, Bernhard
	and Godin, Robert",
	title="Planarity of Lattices",
	booktitle="Formal Concept Analysis",
	year="2005",
	publisher="Springer Berlin Heidelberg",
	address="Berlin, Heidelberg",
	pages="391--402",
	isbn="978-3-540-32262-7"
}

@InProceedings{Zschalig2006,
	author="Zschalig, Christian",
	editor="Missaoui, Rokia
	and Schmidt, J{\"u}rg",
	title="Characterizing Planar Lattices Using Left-Relations",
	booktitle="Formal Concept Analysis",
	year="2006",
	publisher="Springer Berlin Heidelberg",
	address="Berlin, Heidelberg",
	pages="280--290",
	isbn="978-3-540-32204-7"
}

@InProceedings{Zschalig2007,
	author="Zschalig, Christian",
	editor="Kuznetsov, Sergei O.
	and Schmidt, Stefan",
	title="Bipartite Ferrers-Graphs and Planar Concept Lattices",
	booktitle="Formal Concept Analysis",
	year="2007",
	publisher="Springer Berlin Heidelberg",
	address="Berlin, Heidelberg",
	pages="313--327",
	isbn="978-3-540-70901-5"
}

@inproceedings{ZschaligFDP,
	author = {Zschalig, Christian},
	booktitle = {CLA},
	editor = {Eklund, Peter W. and Diatta, Jean and Liquiere, Michel},
	ee = {http://ceur-ws.org/Vol-331/Zschalig.pdf},
	publisher = {CEUR-WS.org},
	series = {CEUR Workshop Proceedings},
	title = {An FDP-Algorithm for Drawing Lattices},
	volume = 331,
	year = 2007
}

@InProceedings{Purchase1997,
	author="Purchase, Helen",
	editor="DiBattista, Giuseppe",
	title="Which aesthetic has the greatest effect on human understanding?",
	booktitle="Graph Drawing",
	year="1997",
	publisher="Springer Berlin Heidelberg",
	address="Berlin, Heidelberg",
	pages="248--261",
	isbn="978-3-540-69674-2"
}

@InProceedings{Purchase2000,
	author="Purchase, Helen C.
	and Allder, Jo-Anne
	and Carrington, David",
	editor="Marks, Joe",
	title="User Preference of Graph Layout Aesthetics: A UML Study",
	booktitle="Graph Drawing",
	year="2001",
	publisher="Springer Berlin Heidelberg",
	address="Berlin, Heidelberg",
	pages="5--18",
	isbn="978-3-540-44541-8"
}

@article{Biehl2017,
	author = {Biehl, Martin},
	year = {2017},
	month = {04},
	pages = {},
	title = {Formal approaches to a definition of agents},
	doi = {10.48550/arXiv.1704.02716}
}

@article{Duerrschnabel2019,
  author       = {Dominik D{\"{u}}rrschnabel and
                  Tom Hanika and
                  Gerd Stumme},
  title        = {DimDraw - {A} novel tool for drawing concept lattices},
  journal      = {CoRR},
  volume       = {abs/1903.00686},
  year         = {2019},
  url          = {http://arxiv.org/abs/1903.00686},
  eprinttype    = {arXiv},
  eprint       = {1903.00686},
  bibsource    = {dblp computer science bibliography, https://dblp.org}
}

@article{Duerrschnabel2023,
  author       = {Dominik D{\"{u}}rrschnabel and
                  Tom Hanika and
                  Gerd Stumme},
  title        = {Drawing Order Diagrams Through Two-Dimension Extension},
  journal      = {CoRR},
  volume       = {abs/1906.06208},
  year         = {2019},
  url          = {http://arxiv.org/abs/1906.06208},
  eprinttype    = {arXiv},
  eprint       = {1906.06208},
  bibsource    = {dblp computer science bibliography, https://dblp.org}
}

@inproceedings{Duerrschnabel2025,
  author    = {D\"{u}rrschnabel, Dominik},
  title     = {odis: A {R}ust {L}ibrary and {W}eb {GUI} for {FCA}},
  series    = {Conceptual Knowledge Software: Recent Advancements and Examples},
  year      = {2025},
  url       = {https://www.kde.cs.uni-kassel.de/consoft/assets/odis.pdf}
}

@incollection{Barrett2009,
	author = {Clark Barrett and Roberto Sebastiani and Sanjit Seshia and Cesare Tinelli},
	editor = {Armin Biere and Marijn J. H. Heule and Hans van Maaren and Toby Walsh},
	title = {Satisfiability Modulo Theories},
	booktitle = {Handbook of Satisfiability},
	series = {Frontiers in Artificial Intelligence and Applications},
	volume = {185},
	chapter = {26},
	pages = {825--885},
	publisher = {IOS Press},
	month = feb,
	year = {2009}
}

@InProceedings{DeMoura2008,
	author="de Moura, Leonardo
	and Bj{\o}rner, Nikolaj",
	editor="Ramakrishnan, C. R.
	and Rehof, Jakob",
	title="Z3: An Efficient SMT Solver",
	booktitle="Tools and Algorithms for the Construction and Analysis of Systems",
	year="2008",
	publisher="Springer Berlin Heidelberg",
	address="Berlin, Heidelberg",
	pages="337--340",
	isbn="978-3-540-78800-3"
}

@INPROCEEDINGS{Amar2005,
	author={Amar, R. and Eagan, J. and Stasko, J.},
	booktitle={IEEE Symposium on Information Visualization, 2005. INFOVIS 2005.}, 
	title={Low-level components of analytic activity in information visualization}, 
	year={2005},
	volume={},
	number={},
	pages={111-117},
	doi={10.1109/INFVIS.2005.1532136}
}

@InProceedings{Mooney2025,
	author =	{Mooney, Gavin J. and Hegemann, Tim and Wolff, Alexander and Wybrow, Michael and Purchase, Helen C.},
	title =	{{Universal Quality Metrics for Graph Drawings: Which Graphs Excite Us Most?}},
	booktitle =	{33rd International Symposium on Graph Drawing and Network Visualization (GD 2025)},
	pages =	{30:1--30:20},
	series =	{Leibniz International Proceedings in Informatics (LIPIcs)},
	ISBN =	{978-3-95977-403-1},
	ISSN =	{1868-8969},
	year =	{2025},
	volume =	{357},
	editor =	{Dujmovi\'{c}, Vida and Montecchiani, Fabrizio},
	publisher =	{Schloss Dagstuhl -- Leibniz-Zentrum f{\"u}r Informatik},
	address =	{Dagstuhl, Germany},
	URL =		{https://drops.dagstuhl.de/entities/document/10.4230/LIPIcs.GD.2025.30},
	URN =		{urn:nbn:de:0030-drops-250162},
	doi =		{10.4230/LIPIcs.GD.2025.30}
}

@misc{NoehreDimFlux,
	author       = {Nöhre, Marcel and
					Dürrschnabel, Dominik and
					Ganter, Bernhard and
					Stumme, Gerd},
	title        = {A Visual Benchmark of DimFlux: Comparison of Line
					Diagrams for Concept Lattices
					},
	month        = mar,
	year         = 2026,
	publisher    = {Zenodo},
	doi          = {10.5281/zenodo.20280441},
	url          = {https://doi.org/10.5281/zenodo.20280441},
}

\end{document}